\newif\ifproceedings\proceedingstrue
\newif\ifsinglecolumn %\singlecolumntrue

\ifsinglecolumn
    \documentclass[a4paper]{article}
    \usepackage{amsthm,amsmath}
    \usepackage[a4paper,hmargin=0.85in,vmargin=0.95in]{geometry}
\else
    \documentclass[conference,a4paper]{IEEEtran}
\fi
\usepackage[hyphens]{url}
\usepackage[pdfborder={0 0 0},bookmarksopen,breaklinks,pageanchor=false]{hyperref}
\usepackage{graphicx}
\usepackage{multirow}
\usepackage{amsfonts,amssymb}
\usepackage[ruled,algosection,noend,linesnumbered]{algorithm2e}
\usepackage{amsthm,amsmath}
\usepackage{siunitx}
\usepackage{booktabs}
\usepackage{xspace}
\usepackage[numbers,sort]{natbib}
\usepackage{xcolor}
\usepackage{xspace}
\usepackage{framed}
\usepackage[capitalize]{cleveref}
\usepackage{tabularx}
\usepackage{tabu}
\usepackage{wrapfig}
\usepackage[all]{xy}
\usepackage{balance}
\SetMathAlphabet{\mathcal}{normal}{OMS}{lmsy}{m}{n}
\SetMathAlphabet{\mathcal}{bold}{OMS}{lmsy}{m}{n}
\SetSymbolFont{largesymbols}{normal}{OMX}{lmex}{m}{n}
\SetSymbolFont{largesymbols}{bold}{OMX}{lmex}{m}{n}
\usepackage{pgf}
\usepackage{tikz}
\usetikzlibrary{arrows,automata}

\tikzset{font=\footnotesize}

\usepackage[latin1]{inputenc}

\usepackage{bussproofs}

\usepackage{booktabs}
\usepackage{siunitx}
\usepackage{pifont}

\usepackage{microtype}

\newcommand\bi{\begin{itemize}}
\newcommand\ei{\end{itemize}}
\newcommand\ben{\begin{enumerate}}
\newcommand\een{\end{enumerate}}

\newcommand\george[1]{\textcolor{red}{}}
\newcommand\dave[1]{\textcolor{red}{}}
\newcommand\bano[1]{\textcolor{red}{}}
\newcommand\mustafa[1]{\textcolor{red}{}}
\newcommand\alberto[1]{\textcolor{red}{}}

\newcommand\sysname{Chainspace\xspace}
\newcommand\sbac{$\mathcal{S}$-BAC\xspace}
\newcommand\cscoin{CSCoin\xspace}

\newcommand\bitcoin{Bitcoin\xspace}
\newcommand\bitcoinng{Bitcoin-NG\xspace}

\newcommand\byzcoin{ByzCoin\xspace}
\newcommand\rscoin{RSCoin\xspace}
\newcommand\elastico{Elastico\xspace}
\newcommand\omniledger{OmniLedger\xspace}
\newcommand\algorand{Algorand\xspace}
\newcommand\bigchaindb{BigchainDB\xspace}
\newcommand\hyperledgerfabric{Hyperledger Fabric\xspace}

\newcommand\type{\mathsf{type}}
\newcommand\types{\mathsf{types}}
\newcommand\id{\mathsf{id}}
\newcommand\procedures{\mathsf{proc}}
\newcommand\csobject{o}
\newcommand\contract{c}
\newcommand\csprocedure{p}

\newcommand\pinputs{\vec{w}}
\newcommand\preferences{\vec{r}}
\newcommand\poutputs{\vec{x}}

\newcommand\lparams{\mathsf{lpar}}
\newcommand\lreturns{\mathsf{lret}}
\newcommand\sparams{\mathsf{spar}}
\newcommand\sreturns{\mathsf{sret}}
\newcommand\dependencies{\mathsf{dep}}

\newcommand\faulty{f}
\newcommand\checker{v}
\newcommand\transaction{T}
\newcommand\shard{\phi}

\newcommand\LP{\textsc{LocalPrepared}}
\newcommand\APC{\textsc{AllPrepared}}
\newcommand\SPA{\textsc{SomePrepared}}
\newcommand\bftinit{BFT-Initiator~}

\newcommand\modsmart{\textsc{Mod-SMaRt~}}
\newcommand\bftsmart{\textsc{bft-SMaRt~}}

\newtheorem{SecThm}{Security Theorem}
\newtheorem{sbacThm}{\sbac Theorem}

\title{\sysname: A Sharded Smart Contracts Platform}

\ifproceedings{
\author{
\IEEEauthorblockN{Mustafa Al-Bassam\IEEEauthorrefmark{1}, 
Alberto Sonnino\IEEEauthorrefmark{1}, 
Shehar Bano\IEEEauthorrefmark{1},
Dave Hrycyszyn\IEEEauthorrefmark{2} 
and George Danezis\IEEEauthorrefmark{1}}
\IEEEauthorblockA{\IEEEauthorrefmark{1} University College London, United Kingdom} 
\IEEEauthorblockA{\IEEEauthorrefmark{2} constructiveproof.com}
}
}\fi

\begin{document}

\IEEEoverridecommandlockouts
\makeatletter\def\@IEEEpubidpullup{9\baselineskip}\makeatother
\IEEEpubid{\parbox{\columnwidth}{Permission to freely reproduce all or part
    of this paper for noncommercial purposes is granted provided that
    copies bear this notice and the full citation on the first
    page. Reproduction for commercial purposes is strictly prohibited
    without the prior written consent of the Internet Society, the
    first-named author (for reproduction of an entire paper only), and
    the author's employer if the paper was prepared within the scope
    of employment.  %\\
    %NDSS '16, 21-24 February 2016, San Diego, CA, USA\\
    %Copyright 2016 Internet Society, ISBN 1-891562-41-X\\
    %http://dx.doi.org/10.14722/ndss.2016.23187
}
\hspace{\columnsep}\makebox[\columnwidth]{}}

\maketitle

\begin{abstract}
\sysname is a decentralized infrastructure, known as a distributed ledger, that supports user defined smart contracts and executes user-supplied transactions on their objects. The correct execution of smart contract transactions is verifiable by all. The system is scalable, by sharding state and the execution of transactions, and using \sbac, a distributed commit protocol, to guarantee consistency. \sysname is secure against subsets of nodes trying to compromise its integrity or availability properties through Byzantine Fault Tolerance (BFT), and extremely high-auditability, non-repudiation and `blockchain' techniques. Even when BFT fails, auditing mechanisms are in place to trace malicious participants. We present the design, rationale, and details of \sysname; we argue through evaluating an implementation of the system about its scaling and other features; we illustrate a number of privacy-friendly smart contracts for smart metering, polling and banking and measure their performance.
\end{abstract}

\section{Introduction}
\label{introduction}
\sysname is a distributed ledger platform for high-integrity and transparent processing of transactions within a decentralized system. Unlike application specific distributed ledgers, such as Bitcoin~\cite{nakamoto2008bitcoin} for a currency, or certificate transparency~\cite{laurie2013certificate} for certificate verification, \sysname offers extensibility though smart contracts, like Ethereum~\cite{wood2014ethereum}. However, users expose to \sysname enough information about contracts and transaction semantics, to provide higher scalability through sharding across infrastructure nodes: our modest testbed of 60 cores achieves 350 transactions per second, as compared with a peak rate of less than 7 transactions per second for \bitcoin over 6K full nodes. Etherium currently processes 4 transactions per second, out of theoretical maximum of 25. Furthermore, our platform is agnostic as to the smart contract language, or identity infrastructure, and supports privacy features through modern zero-knowledge techniques~\cite{bootle2016efficient,danezis2014square}.% or shorter SNARKs~\cite{danezis2014square}.

Unlike other scalable but `permissioned' smart contract platforms, such as \hyperledgerfabric~\cite{cachin2016architecture} or \bigchaindb~\cite{mcconaghy2016bigchaindb}, \sysname aims to be an `open' system: it allows anyone to author a smart contract, anyone to provide infrastructure on which smart contract code and state runs, and any user to access calls to smart contracts. Further, it provides ecosystem features, by allowing composition of smart contracts from different authors. We integrate a value system, named \cscoin, as a system smart contract to allow for accounting between those parties.

However, the security model of \sysname, is different from traditional unpermissioned blockchains, that rely on proof-of-work and global replication of state, such as Ethereum. In \sysname smart contract authors designate the parts of the infrastructure that are trusted to maintain the integrity of their contract---and only depend on their correctness, as well as the correctness of contract sub-calls. This provides fine grained control of which part of the infrastructure need to be trusted on a per-contract basis, and also allows for horizontal scalability.

This paper makes the following contributions:
\begin{itemize}
    \item It presents \sysname, a system that can scale arbitrarily as the number of nodes increase, tolerates byzantine failures, and can be fully and publicly audited.
    %\item The system leverages sharding of state and transaction checking to scale, and only requires nodes to maintain active state to operate correctly -- not the full history of the system.
    \item It presents a novel distributed atomic commit protocol, called \sbac, for sharding generic smart contract transactions across multiple byzantine nodes, and correctly coordinating those nodes to ensure safety, liveness and security properties. 
    \item It introduces a distinction between parts of the smart contract that execute a computation, and those that check the computation and discusses how that distinction is key to supporting privacy-friendly smart-contracts.
    % \item It implements key system features, such as smart contract creation and management, auditing of the infrastructure, and accounting as `system' contracts parameterizing an instantiation of \sysname.
    \item It provides a full implementation and evaluates the performance of the byzantine distributed commit protocol, \sbac, on a real distributed set of nodes and under varying transaction loads.
    \item It presents a number of key system and application smart contracts and evaluates their performance. The contracts for privacy-friendly smart-metering and privacy-friendly polls illustrate and validate support for high-integrity and high-privacy applications.
\end{itemize}

\noindent {\bf Outline:} \Cref{Overview} presents an overview of \sysname; \Cref{Interface} presents the client-facing application interface; \Cref{design} presents the design of internal data structures guaranteeing integrity, the distributed architecture, the byzantine commit protocols, and smart contract definition and composition. \Cref{theorems} argues the correctness and security; specific smart contracts and their evaluations are presented in \Cref{applications}; \Cref{evaluation} presents an evaluation of the core protocols and smart contract performance; \Cref{limits} presents limitation and \Cref{related} a comparison with related work; and \Cref{conclusions} concludes.

\section{System Overview}
\label{Overview}
\sysname allows applications developers to implement distributed ledger applications by defining and calling procedures of smart contracts operating on controlled objects, and abstracts the details of how the ledger works and scales. In this section, we first describe data model of \sysname, followed by an overview of the system design, its threat model and security properties.

\subsection{Data Model: Objects, Contracts, Transactions.}
\label{data-model}

%Developers need to understand the nature of specifying a \sysname service, as well as the constraints that allow for effective and automatics parallel execution, support for privacy, and serialization---and 
\sysname applies aggressively the end-to-end principle~\cite{saltzer1984end} in relying on untrusted end-user applications to build transactions to be checked and executed. We describe below key concepts within the \sysname data model, that developers need to grasp to use the system.

\emph{Objects} are atoms that hold state in the \sysname system. We usually refer to an object through the letter $o$, and a set of objects as $\csobject \in O$. All objects have a cryptographically derived unique identifier used to unambiguously refer to the object, that we denote $\id(\csobject)$. Objects also have a type, denoted as $\type(\csobject)$, that determines the unique identifier of the smart contract that defines them, and a type name. In \sysname object state is immutable. Objects may be in two meta-states, either \emph{active} or \emph{inactive}. Active objects are available to be operated on through smart contract procedures, while inactive ones are retained for the purposes of audit only.

\emph{Contracts} are special types of objects, that contain executable information on how other objects of types defined by the contract may be manipulated. They define a set of initial objects that are created when the contract is first created within \sysname. A contract $\contract$ defines a \emph{namespace} within which \emph{types} (denoted as $\types(\contract)$) and a \emph{checker} $\checker$ for \emph{procedures} (denoted as $\procedures(\contract)$) are defined.

A \emph{procedure}, $\csprocedure$, defines the logic by which a number of objects, that may be \emph{inputs} or \emph{references}, are processed by some logic and \emph{local parameters} and \emph{local return values} (denoted as $\lparams$ and $\lreturns$), to generate a number of object \emph{outputs}. Notionally, input objects, denoted as a vector $\pinputs$, represent state that is invalidated by the procedure; references, denoted as $\preferences$ represent state that is only read; and outputs are objects, or $\poutputs$ are created by the procedure. Some of the local parameters or local returns may be secrets, and require confidentiality. We denote those as $\sparams$ and $\sreturns$ respectively.

We denote the execution of such a procedure as: 
\begin{equation}\label{eq:procedure}
\contract.\csprocedure(\pinputs, \preferences, \lparams, \sparams) \rightarrow \poutputs, \lreturns, \sreturns
\end{equation}
for $\pinputs, \preferences, \poutputs \in O$ and $\csprocedure \in \procedures(\contract)$. We restrict the type of all objects (inputs $\pinputs$, outputs $\poutputs$ and references $\preferences$) to have types defined by the same contract $\contract$ as the procedure $\csprocedure$ (formally: $\forall \csobject \in \pinputs \cup \poutputs \cup \preferences . \type(\csobject) \in \types(\contract)$). However, public locals (both $\lparams$ and $\lreturns$) may refer to objects that are from different contracts through their identifiers. We further require a procedure that outputs an non empty set of objects $\poutputs$, to also take as parameters a non-empty set of input objects $\pinputs$. Transactions that create no outputs are allowed to just take locals and references $\preferences$.

Associated with each smart contract $\contract$, we define a \emph{checker} denoted as $\checker$. Those checkers are pure functions (ie.\ deterministic, and have no side-effects), and return a Boolean value. A checker $\checker$ is defined by a contract, and takes as parameters a procedure $\csprocedure$, as well as inputs, outputs, references and locals.
\begin{equation}\label{eq:checker}
\contract.\checker(\csprocedure, \pinputs, \preferences, \lparams, \poutputs, \lreturns, \dependencies) \rightarrow \{\mathsf{true}, \mathsf{false}\}
\end{equation}
Note that checkers do not take any secret local parameters ($\sparams$ or $\sreturns$). A checker for a smart contract returns $\mathsf{true}$ only if there exist some secret parameters $\sparams$ or $\sreturns$, such that an execution of the contract procedure $\csprocedure$, with the parameters passed to the checker alongside $\sparams$ or $\sreturns$, is possible as defined in \Cref{eq:procedure}.
% However, in case secrets are necessary to generate transactions, the checker needs to perform a fundamentally different operation then the procedure that generates a valid transaction. The simplest instance of this, is a procedure that relies on a digital signature scheme to generate valid signature authorizing the use of an object, on the basis of a secret signature key. The checker for such a procedure, instead of signing using a secret signature key, would simply verify the signatures using the signature verification function, and the associated verification key instead.
The variable $\dependencies$ represent the context in which the procedure is called: namely information about other procedure executions. This supports composition, as we discuss in detail in the next section.

We note that procedures, unlike checkers, do not have to be pure functions, and may be randomized, keep state or have side effects. A smart contract defines explicitly the checker $\contract.\checker$, but does not have to define procedures \emph{per se}. The \sysname system is oblivious to procedures, and relies merely on checkers. Yet, applications may use procedures to create valid transactions. The distinction between procedures and checkers---that do not take secrets---is key to implementing 
%the simplest of contracts (ie.\ supporting authentication though digital signatures) as well as advanced 
privacy-friendly contracts.

\emph{Transactions} represent the atomic application of one or more valid procedures to active input objects, and possibly some referenced objects, to create a number of new active output objects. The design of \sysname is user-centric, in that a user client executes all the computations necessary to determine the outputs of one or more procedures forming a transaction, and provides enough evidence to the system to check the validity of the execution and the new objects.

Once a transaction is accepted in the system it `consumes' the input objects, that become inactive, and brings to life all new output objects that start their life by being active. References on the other hand must be active for the transaction to succeed, and remain active once a transaction has been successfully committed. % Besides input, referenced, and output objects transactions may contain state local to the transaction that is fully encapsulated within the transaction.

%In the language of traditional programming languages, transaction inputs are passed by reference---and may be modified by the transaction; references are constant references that may not be modified by the transactions; and local variables are simply the arguments passed to a function implementing the transaction.

An client packages enough information about the execution of those procedures to allow \sysname to safely \emph{serialize} its execution, and \emph{atomically} commit it only if all transactions are valid according to relevant smart contract checkers. %We discuss the information that users should provide the system to commit or abort transactions next.

% \george{What is the autorization model? In BTC this is parametrized by the code---we could also have authorization logic in the contract or type of object.}

\begin{figure}[!t]
    \centering
    \includegraphics[width=.48\textwidth]{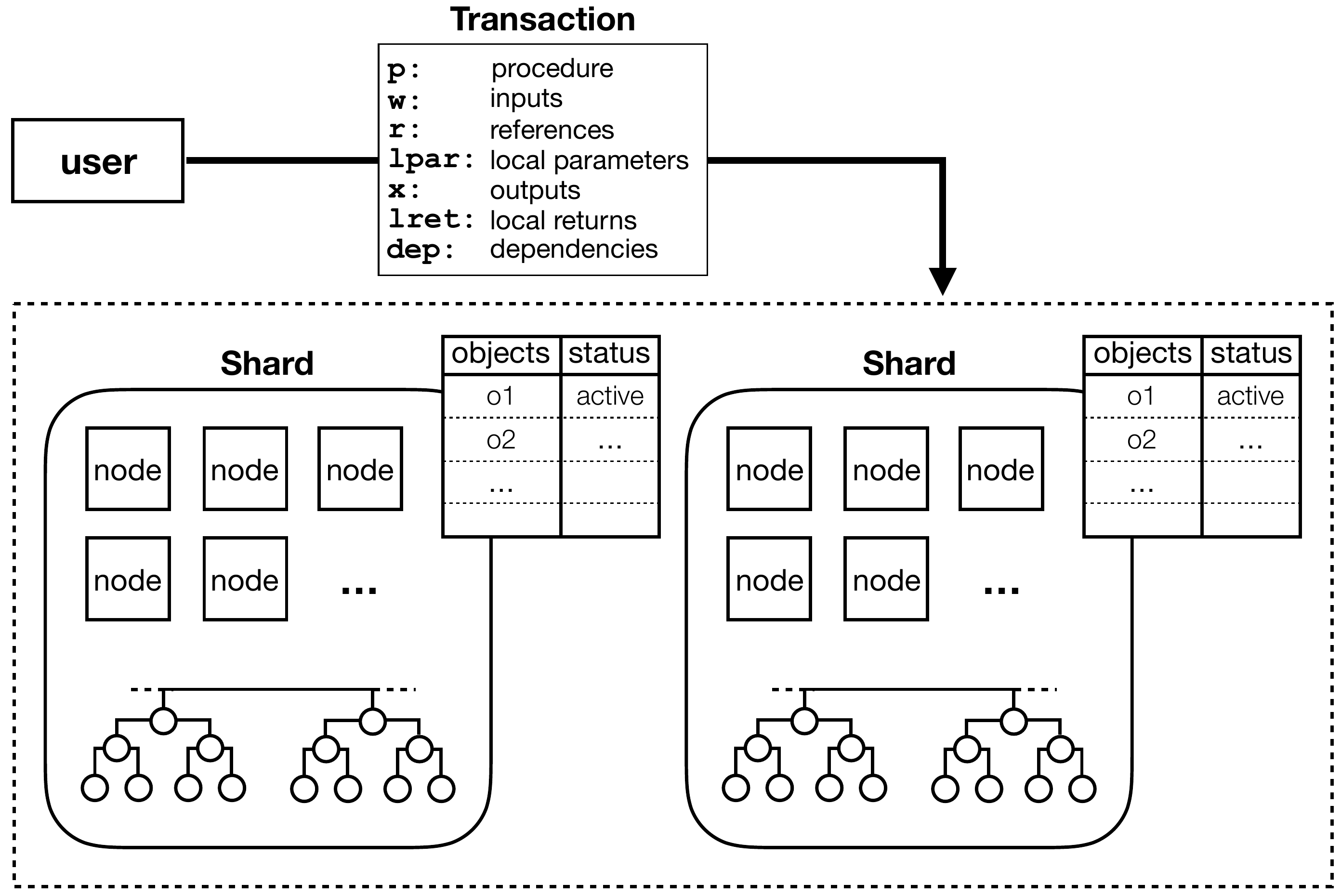}
    \caption{Design overview of \sysname system.}
    \label{fig:design}
\end{figure}

\subsection{System Design, Threat Model and Security Properties}
\label{threat-security}

We provide an overview of the system design, illustrated in \Cref{fig:design}. \sysname is comprised of a network of infrastructure \emph{nodes} that manage valid objects, and ensure that only valid transactions are committed. A key design goal is to achieve scalability in terms of high transaction throughput and low latency. To this end, nodes are organized into shards that manage the state of objects, keep track of their validity, and record transactions aborted or committed. Within each shard all honest nodes ensure they consistently agree whether to accept or reject a transaction: whether an object is active or inactive at any point, and whether traces from contracts they know check. Across shards, nodes must ensure that transactions are \emph{committed} if all shards are willing to commit the transaction, and rejected (or \emph{aborted}) if any shards decide to abort the transaction---due to checkers returning \textsf{false} or objects being inactive. To satisfy these requirements, \sysname implements \sbac---a protocol that composes existing Byzantine agreement and atomic commit primitives in a novel way. Consensus on committing (or aborting) transactions takes place in parallel across different shards. For transparency and auditability, nodes in each shard periodically publish a signed hash chain of \emph{checkpoints}: shards add a block (Merkle tree) of evidence including transactions processed in the current epoch, and signed promises from other nodes, to the hash chain. %The hash chain is extended by taking a hash of the head of the chain so far with the new head, which is then collectively signed by honest peers in the shard.    

\sysname supports security properties against two distinct types of adversaries, both polynomial time bounded:

\begin{itemize}
\item  {\bf Honest Shards (HS).} The first adversary may create arbitrary contracts, and input arbitrary transactions into \sysname, however they are bound to only control up to $\faulty$ faulty nodes in any shard. As a result, and to ensure the correctness and liveness properties of Byzantine consensus, each shard must have a size of at least $3\faulty + 1$ nodes. % Under the assumption of `honest shards' no invalid transaction will ever be accepted.

\item {\bf Dishonest Shards (DS).} The second adversary has, additionally to HS, managed to gain control of one or more shards, meaning that they control over $\faulty$ nodes in those shards. Thus, its correctness or liveness may not be guaranteed. %In such a case, invalid transactions may be accepted, within the shard (even at honest nodes) and across shards. Some instances of \sbac may not terminate---holding locks on objects. Under a dishonest shard model we lose correctness and liveness, however auditability properties hold---dishonest parties or shards may be detected and excluded from the system.
\end{itemize}

Faulty nodes in shards may behave arbitrarily, and collude to violate any of the security, safely or liveness properties of the system. They may emit incorrect or contradictory messages, as well as not respond to any or some requests. 

Given this threat model, \sysname supports the following security properties: 

\begin{itemize}

\item \textbf{Transparency.} \sysname ensures that anyone in possession of the identity of a valid object may authenticate the full history of transactions and objects that led to the creation of the object. No transactions may be inserted, modified or deleted from that causal chain or tree. Objects may be used to self-authenticate its full history---this holds under both the HS and DS threat models.

\item \textbf{Integrity.} Subject to the HS threat model, when one or more transactions are submitted only a set of valid non-conflicting transactions will be executed within the system. This includes resolving conflicts---in terms of multiple transactions using the same objects---ensuring the validity of the transactions, and also making sure that all new objects are registered as active. Ultimately, \sysname transactions are accepted, and the set of active objects changes, as if executed sequentially---however, unlike other systems such as Ethereum~\cite{wood2014ethereum}, this is merely an abstraction and high levels of concurrency are supported.

\item \textbf{Encapsulation.} The smart contract checking system of \sysname enforces strict isolation between smart contracts and their state---thus prohibiting one smart contract from directly interfering with objects from other contracts. Under both the HS and DS threat models. However, cross-contract calls are supported but mediated by well defined interfaces providing encapsulation.

% Discovering whether an object is active or not, is ill defined: since any system statement that the object is active, may race a transaction that uses it, making such a statement unreliable and open to time-of-check time-of-use bugs. However, the \sysname system allows strong guarantees that a transaction was executed, that users can rely upon.

\item \textbf{Non-repudiation.} In case conflicting or otherwise invalid transactions were to be accepted in honest shards (in the case of the DS threat model), then evidence exists to pinpoint the parties or shards in the system that allowed the inconsistency to occur. Thus, failures outside the HS threat model, are detectable; the guildy parties may be banned; and appropriate off-line recovery mechanisms could be deployed.

\end{itemize}

\section{The \sysname Application Interface}
\label{Interface}

\begin{figure*}
    \centering

\begin{prooftree}
\AxiomC{$\alpha_0, \mathit{Valid}(t), \alpha'$}
\AxiomC{$\alpha', \mathit{Valid}(T'), \alpha_1$}
\RightLabel{(Sequence)}
\BinaryInfC{$\alpha_0, \mathit{Valid}(T = t::T'), \alpha_1$}
\end{prooftree}

\begin{prooftree}
\AxiomC{$\alpha_0, \mathit{Valid}(\dependencies), \alpha'$}
\AxiomC{$\alpha', \contract.\checker(\csprocedure, \pinputs, \preferences, \lparams, \poutputs, \lreturns, \dependencies), (\alpha' \setminus \pinputs) \cup \poutputs $}
%\AxiomC{\parbox{5cm}{ \\  \\ }}  
\AxiomC{$\pinputs, \preferences \in \alpha' \land$}
\noLine
\UnaryInfC{$(\poutputs \neq \emptyset) \rightarrow (\pinputs \neq \emptyset) \land$}
\noLine
\UnaryInfC{$\forall \csobject \in \pinputs \cup \poutputs \cup \preferences . \type(\csobject) \in \types(\contract)$}
\RightLabel{(Check)}
\TrinaryInfC{$\alpha_0, \mathit{Valid}(t = [\contract, \csprocedure, \pinputs, \preferences, \poutputs, \lparams, \lreturns, \dependencies]), (\alpha' \setminus \pinputs) \cup \poutputs$}
\end{prooftree}
    \caption{The sequencing and checking validity rules for transactions.}
    \label{fig:rules}
\end{figure*}

Smart Contract developers in \sysname register a smart contract $c$ into the distributed system managing \sysname, by defining a  checker for the contract and some initial objects. Users may then submit transactions to operate on those objects in ways allowed by the checkers. %and attaching them with a contract name space. Then procedures $\contract.\csprocedure$ within the new contract may be invoked to create new active objects the types exported by $\contract$, and other procedures may take active objects, and reference others, to generate new active objects through procedures $\contract.\csprocedure'$ in the smart contract. 
Transactions represent the execution of one or more procedures from one or more smart contracts. It is necessary for all inputs to all procedures within the transaction to be active for a transaction to be executed and produce any output objects. 

Transactions are \emph{atomic}: either all their procedures run, and produce outputs, or none of them do. Transactions are also \emph{consistent}: in case two transactions are submitted to the system using the same active object inputs, at most one of them will eventually be executed to produce outputs. Other transactions, called \emph{conflicting}, will be aborted.

\noindent {\bf Representation of Transactions.} A transaction within \sysname is represented by sequence of \emph{traces} of the executions of the procedures that compose it, and their interdependencies. These are computed and packaged by end-user clients, and contain all the information  a checker needs to establish its correctness. A Transaction is a data structure such that:
\begin{align*}
\text{type}\ &\textit{Transaction}: \textit{Trace}\ \text{list}\\
\text{type}\ &\textit{Trace}: \text{Record}\ \{ \\
&\contract: \id(\csobject), \quad
\csprocedure: \text{string},\\
&\pinputs, \preferences, \poutputs: \id(\csobject)\ \text{list},\\
&\lparams, \lreturns: \text{arbitrary data}, \\
&\dependencies: \textit{Trace}\ \text{list} \}
\end{align*}
To generate a set of traces composing the transaction, a \emph{user executes on the client side all the smart contract procedures} required on the input objects, references and local parameters, and generates the output objects and local returns for every procedure---potentially also using secret parameters and returns. Thus the actual computation behind the transactions is performed by the user, and the traces forming the transaction already contain the output objects and return parameters, and sufficient information to check their validity through smart contract checkers. This design pattern is related to traditional \emph{optimistic concurrency control}.

Only valid transactions are eventually committed into the \sysname system, as specified by two validity rules \emph{sequencing} and \emph{checking} presented in \Cref{fig:rules}. Transactions are considered valid within a context of a set of active objects maintained by \sysname, denoted with $\alpha$. Valid transactions lead to a new context of active objects (eg.\ $\alpha'$). We denote this through the triplet ($\alpha, \textit{Valid}(T), \alpha'$), which is true if the execution of transaction $T$ is valid within the context of active objects $\alpha$ and generates a new context of active objects $\alpha'$. The two rules are as follows:

\begin{itemize}
    \item (Sequence rule). A `\textit{Trace} list' (within a `\textit{Transaction}' or list of dependencies) is valid if each of the traces are valid in sequence (see \Cref{fig:rules} rule for sequencing). Further, the active objects set is updated in sequence before considering the validity of each trace.
    \item (Check rule).  A particular `\textit{Trace}' is valid, if the sequence of its dependencies are valid, and then in the resulting active object context, the checker for it returns $\mathsf{true}$. A further three side conditions must hold: (1) inputs and references must be active; (2) if the trace produces any output objects it must also contain some input objects; and (3) all objects passed to the checker must be of types defined by the smart contract of this checker (see \Cref{fig:rules} rule for checking).
\end{itemize}

The ordering of active object sets in the validation rules result in a depth-first validation of all traces, which represents a depth-first execution and data flow dependency between them. It is also noteworthy that only the active set of objects needs to be tracked to determine the validity of new transactions, which is in the order of magnitude of active objects in the system. The much longer list of inactive objects, which grows to encompass the full history of every object in the system is not needed---which we leverage to enable better when validating transactions. It also results in a smaller amount of working memory to perform incremental audits.

%A transaction is valid if a sequential checking of all its constituent traces are valid. Each trace is valid, if the checker associated with the procedure it represents returns true, and if the input and reference objects it relies upon are active. 

A valid transaction is executed in a serialized manner, and committed or aborted atomically. If it is committed, the new set of active objects replaces the previous set; if not the set of active objects does not change. Determining whether a transaction may commit involves ensuring all the input objects are active, and all are consumed as a result of the transaction executing, as well as all new objects becoming available for processing (references however remain active). \sysname ensures this through the distributed atomic commit protocol, \sbac.

% A transaction is valid, if all the inputs to all the procedures it calls are active; and if the combination of inputs, references, and local transaction state, applied to the procedures called yields the declared transaction outputs---as checked by the appropriate contract `checker'. Since procedures and checkers are deterministic, all input and referenced objects are immutable, and the local state is bundled with the transaction, anyone may verify the validity of a transaction outputs---up to whether the inputs are active or not.

%Determining whether a transaction may proceed also involved ensuring all the input objects are active, and all are consumed as a result of the transaction executing, as well as all new objects becoming available for processing. \sysname ensures this through distributed consensus and locking protocols.

%\subsection{Smart Contract Composition}
%\mustafa{Using a call-stack model for cross-contract calls means that the batching system below is independent of cross-contract calls. I'm writing a section here and keeping the below section temporarily.}

\noindent {\bf Smart contract composition.} A contract procedure may call a transaction of another smart contract, with specific parameters and rely upon returned values. This is achieved through passing the $\dependencies$ variable to a smart contract checker, a validated list of traces of all the sub-calls performed. The checker can ensure that the parameters and return values are as expected, and those dependencies are checked for validity by \sysname.

Composition of smart contracts is a key feature of a transparent and auditable computation platform. It allows the creation of a library of smart contracts that act as utilities for other higher-level contracts: for example, a simple contract can implement a cryptographic currency, and other contracts---for e-commerce for example---can use this currency as part of their logic. Furthermore, we compose smart contracts, in order to build some of the functionality of \sysname itself as a set of `system' smart contracts, including management of shards mapping to nodes, key management of shard nodes, and governance.

\sysname also supports the atomic batch execution of multiple procedures for efficiency, that are not dependent on each other.

%\mustafa{I need to talk about top-level transaction batching in more detail with respect to the below paragraphs.}

\noindent {\bf Reads.} Besides executing transactions, \sysname clients, need to read the state of objects, if anything, to correctly form transactions. Reads, by themselves, cannot lead to inconsistent state being accepted into the system, even if they are used as inputs or references to transactions. This is a result of the system checking the validity rules before accepting a transaction, which will reject any stale state. 

Thus, any mechanism may be used to expose the state of objects to clients, including traditional relational databases, or `no-SQL' alternatives. Additionally, any indexing mechanism may be used to allow clients to retrieve objects with specific characteristics faster. Decentralized, read-only stores have been extensively studied, so we do not address the question of reads further in this work.

\noindent {\bf Privacy by design.} Defining smart contract logic as checkers allows \sysname to support privacy friendly-contracts by design. In such contracts some information in objects is not in the clear, but instead either encrypted using a public key, or committed using a secure commitment scheme as~\cite{pedersen1991non}. The transaction only contains a valid proof that the logic or invariants of the smart contract procedure were applied correctly or hold respectively, and can take the form of a zero-knowledge proof, or a Succinct Argument of Knowledge (SNARK). Then, generalizing the approach of~\cite{miers2013zerocoin}, the checker runs the verifier part of the proof or SNARK that validates the invariants of the transactions, without revealing the secrets within the objects to the verifiers.

\section{The \sysname System Design}
\label{design}
In \sysname a network of infrastructure \emph{nodes} manages valid objects, and ensure key invariants: namely that only valid transactions are committed. We discuss the data structures nodes use collectively and locally to ensure high integrity; and the distributed protocols they employ to reach consensus on the accepted transactions.

\subsection{High-Integrity Data Structures}

\sysname employs a number of high-integrity data structures. They enable those in possession of a valid object or its identifier
to verify all operations that lead to its creation; they are also used to support \emph{non-equivocation}---preventing \sysname nodes from providing a split view of the state they hold without detection.

\noindent {\bf Hash-DAG structure.} Objects and transactions naturally form a directed acyclic graph (DAG): given an initial state of active objects a number of transactions render their inputs invalid, and create a new set of outputs as active objects. These may be represented as a directed graph between objects, transactions and new objects and so on. Each object may only be created by a single transaction trace, thus cycles between future transactions and previous objects never occur. We prove that output object identifiers resulting from valid transactions are fresh (see \Cref{Thm1}). Hence, the graph of objects inputs, transactions and objects outputs form a DAG, that may be indexed by their identifiers.

We leverage this DAG structure, and augment it to provide a high-integrity data structure. Our principal aim is to ensure that given an object, and its identifier, it is possible to unambiguously and unequivocally check all transactions and previous (now inactive) objects and references that contribute to the existence of the object. To achieve this we define as an identifier for all objects and transactions a cryptographic hash that directly or indirectly depends on the identifiers of all state that contributed to the creation of the object.

Specifically, we define a function $\id(\mathit{Trace})$ as the identifier of a trace contained in transaction $T$. The identifier of a trace is a cryptographic hash function over the name of contract and the procedure producing the trace; as well as serialization of the input object identifiers, the reference object identifiers, and all local state of the transaction (but not the secret state of the procedures); the identifiers of the trace's dependencies are also included. Thus all information contributing to defining the Trace is included in the identifier, except the output object identifiers.

We also define the $\id(\csobject)$ as the identifier of an object $\csobject$. We derive this identifier through the application of a cryptographic hash function, to the identifier of the trace that created the object $\csobject$, as well as a unique name assigned by the procedures creating the trace, to this output object. (Unique in the context of the outputs of this procedure call, not globally, such as a local counter.) % Transactions that assign the same name to two outputs are considered invalid.

An object identifier $\id(\csobject)$ is a high-integrity handle that may be used to authenticate the full history that led to the existence of the object $\csobject$. Due to the collision resistance properties of secure cryptographic hash functions an adversary is not able to forge a past set of objects or transactions that leads to an object with the same identifier. Thus, given  $\id(\csobject)$ anyone can verify the authenticity of a trace that led to the existence of $\csobject$.

A very important property of object identifiers is that future transactions cannot re-create an object that has already become inactive. Thus checking object validity only requires maintaining a list of active objects, and not a list of past inactive objects:
\begin{SecThm}\label{Thm1}
No sequence of valid transactions, by a polynomial time constrained adversary, may re-create an object with the same identifier with an object that has already been active in the system.
\end{SecThm}
\begin{proof}
\small
We argue this property by induction on the serialized application of valid transactions, and for each transaction by structural induction on the two validity rules. Assuming a history of $n-1$ transactions for which this property holds we consider transaction $n$. Within transaction $n$ we sequence all traces and their dependencies, and follow the data flow of the creation of new objects by the `check' rule. For two objects to have the same $\id(\csobject)$ there need to be two invocations of the check rule with the same contract, procedure, inputs and references. However, this leads to a contradiction: once the first trace is checked and considered valid the active input objects are removed from the active set, and the second invocation becomes invalid. Thus, as long as object creation procedures have at least one input (which is ensured by the side condition) the theorem holds, unless an adversary can produce a hash collision. The inductive base case involves assuming that no initial objects start with the same identifier -- which we can ensure axiomatically.
\end{proof}

%By the same properties of hash cryptographic hash functions, it impossible for an adversary to find transactions or objects that may lead to an object with the same identifier with another active, or inactive object. This ensure the directed and acyclic structure of the DAG is maintained without any additional checks beyond ensuring the identifiers of objects and transactions are derived correctly at each transaction.

We call this directed acyclic graph with identifiers derived using cryptographic functions a Hash-DAG, and we make extensive use of the identifiers of objects and their properties in \sysname.

\noindent {\bf Node Hash-Chains.} Each node in \sysname, that is entrusted with preserving integrity, associates with its shard a hash chain. Periodically, peers within a shard consistently agree to seal a \emph{checkpoint}, as a block of transactions into their hash chains. They each form a Merkle tree containing all transactions that have been accepted or rejected in sequence by the shard since the last checkpoint was sealed. Then, they extend their hash chain by hashing the root of this Merkle tree and a block sequence number, with the head hash of the chain so far, to create the new head of the hash chain. Each peer signs the new head of their chain, and shares it with all other peers in the shard, and anyone who requests it. For strong auditability additional information, besides committed or aborted transactions, has to be included in the Merkle tree: node should log any promise to either commit or abort a transaction from any other peer in any shard (the prepared(\transaction,*) statements explained in the next sections). 

All honest nodes within a shard independently create the same chain for a checkpoint, and a signature on it---as long as the consensus protocols within the shards are correct. We say that a checkpoint represents the decision of a shard, for a specific sequence number, if at least $\faulty + 1$ signatures of shard nodes sign it. On the basis of these hash chains we define a \emph{partial audit} and a \emph{full audit} of the \sysname system.

In a \emph{partial audit} a client is provided evidence that a transaction has been either committed or aborted by a shard. A client performing the partial audit may request from any node of the shard evidence for a transaction \transaction. The shard peer will present a block representing the decision of the shard, with $\faulty + 1$ signatures, and a proof of inclusion of a commit or abort for the transaction, or a signed statement the transaction is unknown. A partial audit provides evidence to a client of the fate of their transaction, and may be used to detect past of future violations of integrity. A partial audit is an efficient operation since the evidence has size $O(s + \log N)$ in $N$ the number of transactions in the checkpoint and $s$ the size of the shard---thanks to the efficiency of proving inclusion in a Merkle tree, and checking signatures.

A \emph{full audit} involves replaying all transactions processed by the shard, and ensuring that (1) all transactions were valid according to the checkers the shard executed; (2) the objects input or references of all committed transactions were all active (see rules in \Cref{fig:rules}); and (3) the evidence received from other shards supports committing or aborting the transactions. To do so an auditor downloads the full hash-chain representing the decisions of the shard from the beginning of time, and re-executes all the transactions in sequence. This is possible, since---besides their secret signing keys---peers in shards have no secrets, and their execution is deterministic once the sequence of transactions is defined. Thus, an auditor can re-execute all transactions in sequence, and check that their decision to commit or abort them is consistent with the decision of the shard. Doing this, requires any inter-shard communication (namely the promises from other shards to commit or abort transactions) to be logged in the hash-chain, and used by the auditor to guide the re-execution of the transactions. A full audit needs to re-execute all transactions and requires evidence of size $O(N)$ in the number $N$ of transactions. This is costly, but may be done incrementally as new blocks of shard decisions are created. 

% Those events may be candidate transactions that take the active object as an input; they may be the state of other peers relating to that object, or they may be a note about which transactions that may affect the object were committed by all peers---some of which we will discuss in the next section in detail.

% The hash-map associated with each object is composed of a sequence of hash tree (Merkle-Tree) roots in blocks, each storing of a set of statements indexed by the object they concern as well as the hash of the previous block. The first \emph{genesis} block is implicitly just records the creation of the hash-map for a peer. This structure builds a hash-chain of roots, and the hash of the latest block authenticates the full history of the chain, and the statements included within it.

% To provide strong forms of non-equivocation within each block we store the root of a Merke Tree representing a hash-map. A hash map represents a non-equivocable set of key-value pairs. Given the root of the tree any auditor can ensure that a particular key is mapped to a specific value, or no value at all. No other value may be stored for a specific key within the structure, and that can be verified efficiently using the root of the tree. In \sysname each peer managing objects logs the state of active or inactive objects, such as whether it is active, whether there are transactions pending for it, or whether it has been consumed by a transaction, within the hash map, keyed by the $\mathsf{ID}(o)$ of the object.

\subsection{Distributed Architecture \& Consensus}

A network of \emph{nodes} manages the state of \sysname objects, keeps track of their validity, and record transactions that are seen or that are accepted as being committed. 

\sysname uses sharding strategies to ensure scalability: a public function $\shard(\csobject)$ maps each object $\csobject$ to a set of nodes, we call a \emph{shard}. These nodes collectively are entrusted to manage the state of the object, keep track of its validity, record transactions that involve the object, and eventually commit at most one transaction consuming the object as input and rendering it inactive. However, nodes must only record such a transaction as committed if they have certainty that all other nodes have, or will in the future, record the same transaction as consuming the object. We call this distributed algorithm the \emph{consensus} algorithm within the shard.

%\sysname clients package transactions, including all the inputs and local variables associated with them, and are responsible for submitting them to the \emph{nodes}. A transaction is valid, and could be potentially accepted by an honest node, under certain conditions:
%\begin{itemize}
%\item Correctness. Given all input objects, referenced object and local state, the transaction when executed must always yield all the specified output objects. This means that the computations of the identifiers of the output objects should be correctly computed and unique, and the content of the output objects should be the result of a faithful execution of the code representing the transaction logic.
%\item Freshness. All the input objects and references are active, and have not already been consumed by another transaction. Furthermore if a peer accepts this transaction as valid, it needs to ensure all other peers will also eventually accept this transaction as valid, thus requiring consensus.
%\end{itemize}

For a transaction $\transaction$ we define a set of \emph{concerned nodes}, $\Phi(\transaction)$ for a transaction structure $\transaction$. We first denote as $\zeta$ the set of all objects identifiers that are input into or referenced by any trace contained in $\transaction$. We also denote as $\xi$ the set of all objects that are output by any trace in $\transaction$. The function $\Phi(\transaction)$ represents the set of nodes that are managing objects that should exist, and be active, in the system for $\transaction$ to succeed. More mathematically, $\Phi(\transaction) = \bigcup \{ \phi(\csobject_i) | \csobject_i \in \zeta \setminus \xi \} $, where $\zeta \setminus \xi$ represents the set of objects input but not output by the transaction itself (its free variables). The set of concerned peers thus includes all shard nodes managing objects that already exist in \sysname that the transaction uses as references or inputs. 

An important property of this set of nodes holds, that ensures that all smart contracts involved in a transaction will be mapped to some concerned nodes that manage state from this contract:

\begin{SecThm}
\label{thm:partition}
If a contract $\contract$ appears in any trace within a transaction $\transaction$, then the concerned nodes set $\Phi(\transaction)$ will contain nodes in a shard managing an object $\csobject$ of a type from contract $\contract$. I.e.\ $\exists o. \type(o) \in \types(\contract) \land \shard(o) \cap \Phi(\transaction) \neq \emptyset$.
\end{SecThm}
\begin{proof}
\small
Consider any trace $t$ within $\transaction$, from contract $c$. If the inputs or references to this trace are not in $\xi$---the set of objects that were created within $\transaction$---then their shards will be included within $\Phi(\transaction)$. Since those are of types within $\contract$ the theorem holds. If on the other hand the inputs or references are in $\xi$, it means that there exists another trace within $\transaction$ from the same contract $\contract$ that generated those outputs. We then recursively apply the case above to this trace from the same $\contract$. The process will terminate with some objects of types in $\contract$ and shard managing them within the concerned nodes set---and this is guarantee to terminate due to the Hash-DAG structure of the transactions (that may have no loops).
\end{proof}

\Cref{thm:partition} ensures that the set of concerned nodes, includes nodes that manage objects from all contracts represented in a transaction. \sysname leverages this to distribute the process of rule validation across peers in two ways:
\begin{itemize}
\item For any existing object $\csobject$ in the system, used as a reference or input within a transaction $\transaction$, only the shard nodes managing it, namely in $\shard(\csobject)$, need to check that it is active (as part of the `check' rule in \Cref{fig:rules}).
\item For any trace $t$ from contract $\contract$ within a transaction $\transaction$, only shards of concerned nodes that manage objects of types within $\contract$ need to run the checker of that contract to validate the trace (again as part of the `check' rule), and that all input, output and reference objects are of types within $\contract$.
\end{itemize}
However, all shards containing concerned nodes for $\transaction$ need to ensure that all others have performed the necessary checks before committing the transaction, and creating new objects.

There are many options for ensuring that concerned nodes in each shards do not reach an inconsistent state for the accepted transactions, such as Nakamoto consensus through proof-of-work~\cite{nakamoto2008bitcoin}, two-phase commit protocols~\cite{lampson1994distributed}, and classical consensus protocols like Paxos~\cite{lamport2001paxos}, PBFT~\cite{castro1999practical}, or xPaxos~\cite{liu2015xft}. However, these approaches lack in performance, scalability, and/or security. We design an open, scalable and decentralized mechanism to perform \emph{Sharded Byzantine Atomic Commit} or \sbac.

\subsection{Sharded Byzantine Atomic Commit (\sbac).}

%\sysname entrusts each object to a shard of nodes, that keep track of whether it exists, it is active or inactive. Within each shard all honest nodes must ensure they consistently agree whether to accept or reject a transaction: whether an object is active or inactive at any point, and whether traces from contracts they know check. Across shards, nodes must ensure that transactions are \emph{committed} if all shards are willing to commit the transaction, and rejected (or \emph{aborted}) if any shards decide to abort the transaction---due to it being invalid or objects involved being inactive. 

\sysname implements the previously described intra-shard consensus algorithm for transaction processing in the \emph{byzantine} and \emph{asynchronous} setting, through the \emph{Sharded Byzantine Atomic Commit} (\sbac) protocol, that combines two primitive protocols: \emph{Byzantine Agreement} and \emph{atomic commit}. 

\begin{itemize}
    \item \emph{Byzantine agreement} ensures that all honest members of a shard of size $3 \faulty + 1$, agree on a specific common sequence of actions, despite 
    some $\faulty$ malicious nodes within the shard. It also guarantees that when agreement is sought, a decision or sequence will eventually be agreed upon. The agreement protocol is executed within each shard to coordinate all nodes. We use \modsmart~\cite{modsmart} implementation of PBFT for state machine replication that provides an optimal number of communications steps (similar to PBFT~\cite{castro1999practical}). This is achieved by replacing reliable broadcast with a special leader-driven Byzantine consensus primitive called Validated and Provable Consensus (VP-Consensus). 
    \item \emph{Atomic commit} is ran across all shards managing objects relied upon by a transaction. It ensures that each shard needs to accept to commit a transaction, for the transaction to be committed; even if a single shard rejects the transaction, then all agree it is rejected. We propose the use of a simple two-phase commit protocol~\cite{bernstein1987rrency}, composed with an agreement protocol to achieve this---loosely inspired by Lamport and Gray~\cite{gray2006consensus}. This protocol was the first to reconcile the needs for distributed commit, and replicated consensus (but only in the non-byzantine setting). % \bano{George: Explain why we choose these protocols?}
\end{itemize}

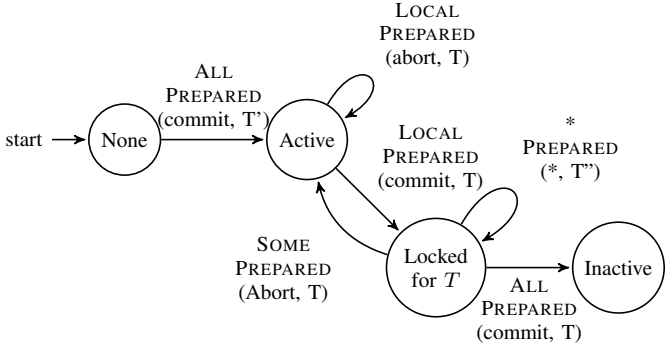
\begin{figure}
    \centering
    
    \begin{tikzpicture}[->,>=stealth',shorten >=1pt,auto,node distance=2.4cm,
                    semithick]
  \tikzstyle{every state}=[fill=white,draw=black,text=black]

  \node[initial,state] (A)                    {None};
  \node[state]         (B) [ right of=A] {Active};
  \node[state]         (C) [below right of=B,align=center] {Locked\\ for $\transaction$};
  \node[state]         (D) [ right of=C] {Inactive};
  
%  \node[state]         (D) [below right of=A] {$q_d$};
%  \node[state]         (C) [below right of=B] {$q_c$};
%  \node[state]         (E) [below of=D]       {$q_e$};

  \path (A) edge [align=center]     node {\textsc{All} \\ \textsc{Prepared} \\ (commit, T')} (B)
%            edge              node {1,1,R} (C)
        (B) edge [in=30,out=60, loop,  align=center] node {\textsc{Local} \\ \textsc{Prepared} \\ (abort, T)} (B)
            edge             [align=center] node {\textsc{Local} \\ \textsc{Prepared} \\ (commit, T)} (C)
        (C) edge             [align=center] node [below] {\textsc{All} \\ \textsc{Prepared} \\ (commit, T)} (D)
            edge [bend left, align=center]  node {\textsc{Some} \\ \textsc{Prepared} \\ (Abort, T)} (B)
            edge [in=30,out=60, loop, align=center]  node {\textsc{*} \\ \textsc{Prepared} \\ (*, T'')} (B)
            ;
%        (D) edge [loop below] node {1,1,R} (D)
%            edge              node {0,1,R} (A)
%        (E) edge [bend left]  node {1,0,R} (A);
\end{tikzpicture}

    \caption{The state machine representing the active, locked and inactive states for any object within \sysname. Each node in a shard replicates the state of the object, and participates in a consensus protocol that allows it to derive the invariants ``Local prepared'', ``All prepared'', and ``Some prepared'' to update the state of an object.}
    \label{fig:state}
\end{figure}

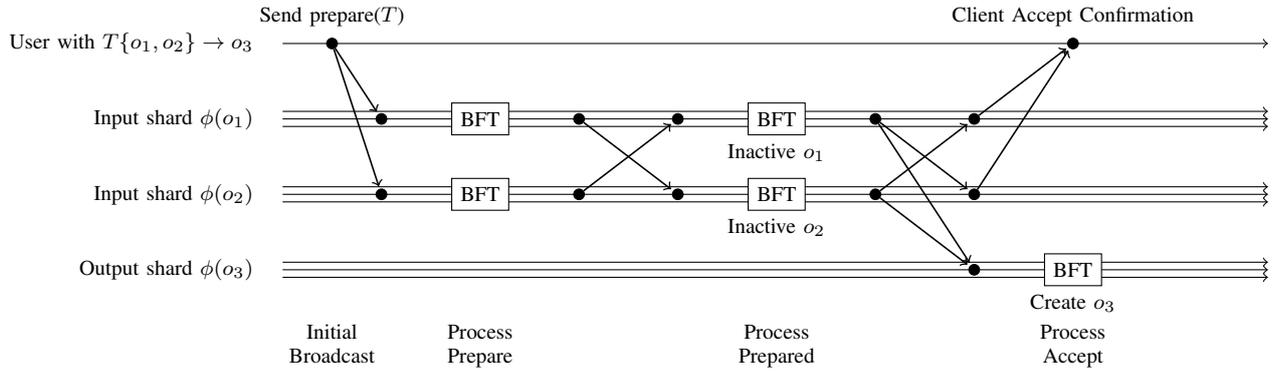
\begin{figure*}[t]
\centering
\begin{tikzpicture}[shorten >=1pt]
 \tikzset{dot/.style={circle,fill=#1,inner sep=0,minimum size=4pt}}

 \path[draw, ->] (0, 3) -- (13, 3) node[left, pos=-0.02] {User with $T \{o_1, o_2\} \rightarrow o_3$}
    node[pos=0.05, dot=black,draw] (ustart) {}
%    node[pos=0.2, dot=black,draw] (upromise) {}
    %node[pos=0.5, dot=black,draw] (u5) {}
%    node[pos=0.45, dot=black,draw] (uend) {}
    node[pos=0.8, dot=black,draw] (ufinal) {};

 \path[draw, ->] (0, 2.1) -- (13, 2.1) node[left, pos=0] {};
 \path[draw, ->] (0, 1.9) -- (13, 1.9) node[left, pos=0] {};
 \path[draw, ->] (0, 2) -- (13, 2) node[left, pos=-0.02] {Input shard $\shard(o_1)$}
    node[pos=0.1, dot=black,draw] (s1p1) {}
    node[rectangle, fill=white, pos=0.2,draw] (s1bft) {BFT}
    node[pos=0.3, dot=black,draw] (s1e1) {}
    node[pos=0.4, dot=black,draw] (s1p2) {}
    node[rectangle, fill=white, pos=0.5,draw] (s1bft2) {BFT}
    node[pos=0.6, dot=black,draw] (s1e2) {}
    node[pos=0.7, dot=black,draw] (s1ee) {};

 \path[draw, ->] (0, 1.1) -- (13, 1.1) node[left, pos=0] {};
 \path[draw, ->] (0, 0.9) -- (13, 0.9) node[left, pos=0] {};
 \path[draw, ->] (0, 1) -- (13, 1) node[left, pos=-0.02] {Input shard $\shard(o_2)$}
    node[pos=0.1, dot=black,draw] (s2p1) {}
    node[rectangle, fill=white, pos=0.2, draw] (s2bft) {BFT}
    node[pos=0.3, dot=black,draw] (s2e1) {}
    node[pos=0.4, dot=black,draw] (s2p2) {}
    node[rectangle, fill=white, pos=0.5, draw] (s2bft2) {BFT}
    node[pos=0.6, dot=black,draw] (s2e2) {}
    node[pos=0.7, dot=black,draw] (s2ee) {};

 \path[draw, ->] (0, 0.1) -- (13, 0.1) node[left, pos=0] {};
 \path[draw, ->] (0, -0.1) -- (13, -0.1) node[left, pos=0] {};
 \path[draw, ->] (0, 0) -- (13, 0) node[left, pos=-0.02] {Output shard $\shard(o_3)$}
    node[pos=0.7, dot=black,draw] (s3p2) {}
    node[rectangle, fill=white, pos=0.8, draw] (s3bft2) {BFT}
    node[pos=0.9] (s3e2) {};

 \path[draw=none] (0, -0.5) -- (13, -0.5) node[left, pos=0] {}
    node[pos=0.05] (Lustart) {}
    % node[pos=0.15, dot=black,draw] (Lp1) {}
    % node[rectangle, fill=white, pos=0.3,draw] (Lbft) {BFT}
    node[pos=0.2] (Le1) {}
    node[pos=0.5] (Lp2) {}
    % node[rectangle, fill=white, pos=0.6,draw] (Lbft2) {BFT}
    node[pos=0.8] (Le2) {};

 \node[below=0.1cm, align=center] at (Lustart) {Initial\\Broadcast};
 %\node[below=0.1cm, align=center] at (Lp1) {Sequence\\Prepare};
 \node[below=0.1cm, align=center] at (Le1) {Process\\Prepare};
 \node[below=0.1cm, align=center] at (Lp2) {Process\\Prepared};
 \node[below=0.1cm, align=center] at (Le2) {Process\\Accept};

\node[above=0.1cm, align=center] at (ustart) {Send prepare($\transaction$)};
%\node[above=0.1cm, align=center] at (upromise) {Receive  \\ promise for $\transaction$};
%\node[above=0.1cm, align=left] at (uend) {Determine:\\ \APC(accept, T) or \\ \SPA(abort, T)};
\node[above=0.1cm, align=left] at (ufinal) {Client Accept Confirmation};

 \node[below=0.2cm] at (s3bft2) {Create $o_3$};
 \node[below=0.2cm] at (s1bft2) {Inactive $o_1$};
 \node[below=0.2cm] at (s2bft2) {Inactive $o_2$};

% \node[above left] at (nd3) {$[1, 2, 0]$};
% \node[below] at (nd4) {$[1, 2, 1]$};
% \node[above] at (nd5) {$[2, 2, 0]$};
% \node[below right] at (nd6) {$[1, 2, 2]$};
% \node[above left] at (nd7) {$[1, 3, 2]$};
% \node[above] at (nd8) {$[3, 3, 2]$};

 \path[semithick, ->] (ustart) edge (s1p1);
 \path[semithick, ->] (ustart) edge (s2p1);

% \path[semithick, ->, dashed] (s1p1) edge (upromise);
% \path[semithick, ->, dashed] (s2p1) edge (upromise);

 \path[semithick, ->] (s1e1) edge (s2p2);
 \path[semithick, ->] (s1e2) edge (s3p2);
 \path[semithick, ->] (s1e2) edge (s2ee);
% \path[semithick, ->, dashed] (s1e1) edge (uend);
 \path[semithick, ->] (s2e1) edge (s1p2);
 \path[semithick, ->] (s2e2) edge (s3p2);
 \path[semithick, ->] (s2e2) edge (s1ee);
% \path[semithick, ->, dashed] (s2e1) edge (uend);
 \path[semithick, ->] (s2ee) edge (ufinal);
 \path[semithick, ->] (s1ee) edge (ufinal);
% \path[semithick, ->] (s3e2) edge (ufinal);

\end{tikzpicture}
\caption{\sbac for a transaction $T$ with two inputs ($o_1, o_2$) and one output object ($o_3$). The user sends the transaction to all nodes in shards managing $o_1$ and $o_2$. The \bftinit takes the lead in sequencing $\transaction$, and emits 'prepared(accept, T)' or 'prepared(abort, T)' to all nodes within the shard. Next the \bftinit of each shard assesses whether overall `All proposed(accept, T)' or `Some proposed(abort, T)' holds across shards, sequences the accept(T,*), and sends the decision to the user. All cross-shard arrows represent a multicast of all nodes in one shard to all nodes in another.}
    \label{fig:net}

\end{figure*}

\sbac composes the above primitives in a novel way to ensure that shards process safely and consistently all transactions. \Cref{fig:net} illustrates a simple example of the \sbac protocol to commit a single transaction with two inputs and one output that we may use as an example. The corresponding object state transitions have been illustrated in \Cref{fig:state}. The combined protocol has been described below. For ease of understanding, in our description we state that all messages are sent and processed by shards. In reality, some of these are handled by a designated node in each shard---the \bftinit---as we discuss at the end of this section.        

    \vspace{2mm} \noindent {\bf Initial Broadcast (Prepare)}. A user acts as a transaction initiator, and sends `prepare(T)' to at least one honest concerned node for transaction $\transaction$. 
    %through a reliable point to point message. 
    To ensure at least one honest node receives it, the user may send the message to $\faulty+1$ nodes of a single shard, or $\faulty+1$ nodes in each concerned shard. % Thus, when a peer has the transaction delivered it can be assured that all other honest concerned peers will also receive the proposed transaction $T$.
     
    \vspace{2mm} \noindent {\bf Sequence Prepare}. Upon a message `prepare(T)' being received, nodes in each shard interpret it as the initiation of a two-phase commit protocol performed across the concerned shards. 
    %Peers within each shard, collectively take on the role of a Resource Manager for the objects the shard manages. 
    The shard locally sequences `prepare(T)' message through the Byzantine consensus protocol. 
    
    \vspace{2mm} \noindent {\bf Process Prepare}. Upon the first action `prepare($T$)' being sequenced through BFT consensus in a shard, nodes of the shard implicitly decide whether it should be committed or aborted. Since all honest nodes in the shard have a consistent replica of the full sequence of actions, they will all decide the same consistent action following `prepare(T)'. 
    
    Transaction $\transaction$ is to be committed if it is valid according to the usual rules (see \Cref{fig:rules}), in brief: (1) the objects input or referenced by $\transaction$ in the shard are active, (2) there is no other instance of the two-phase commit protocol on-going concerning any of those objects (no locks held) and (3) if $\transaction$ is valid according to the validity rules, and the smart contract checkers in the shard. Only the checkers for types of objects held by the shard are checked by the shard.
    
    If the decision is to commit, the shard broadcasts to all concerned nodes `prepared($T$,commit)', otherwise it broadcasts `prepared($T$, abort)'---along with sufficient signatures to convince any party of the collective shard decision (we denote this action as \LP(*, T)). The objects used or referenced by $T$ are `locked' (\Cref{fig:state}) in case of a `prepared commit' until an `accept' decision on the transaction is reached, and subsequent transactions concerning them will be aborted by the shard. Any subsequent `prepare($T''$)' actions in the sequence are ignored, until a matching accept($T$, abort) is reached to release locks, or forever if the transaction is committed.
    
    \vspace{2mm} \noindent {\bf Process Prepared (accept or abort)}. 
    %Each node in a shard listens for `prepared($T$,commit)' or `prepared($T$,abort)' messages from other shard nodes. When it receives $\faulty+1$ messages from nodes in another shard that either all commit or abort, it considers this the authoritative decision of the other shard with respect to the transaction $T$ (we denote this action as \LP(*, T)). Upon receiving a `prepared($T$,commit)' from all shards, concerning transaction $T$, nodes in a shard can decide to.
    Depending on the decision of `prepare($T$)', the shard sequences `accept($T$,commit)' or `accept($T$,abort)'  through the atomic commit protocol across all the concerned shards---along with all messages and signatures of the bundle of `prepared' messages relating to $T$ proving to other shards that the decision should be `accept($T$,commit)' or `accept($T$,abort)' according to its local consensus. 
    If it receives even a single `\LP($T$,abort)' from another shard it instead will move to reach consensus on `accept($T$, abort)' (denoted as \SPA(abort,T)). Otherwise, if all the shards respond with `\LP($T$,commit)' it will reach a consensus on \APC(commit,T). The final decision is sent to the user, along with all messages and signatures of the bundle of `accept' messages relating to $T$ proving that the final decision should be to commit or abort according to responses from all concerned shards.   
    
    %The user mirrors the logic of shards by listening to the `prepared' messages of all input or reference shards and can also decide if the decision will eventually be `accept($T$,commit)' or `accept($T$,abort)'. 
    It is possible, that a shard hears a prepared message for $\transaction$ before a prepare message, due to unreliability, asynchrony or a malicious user. In that case the shard assumes that a `prepare(T)' message is implicit, and sequences it.
    
    \vspace{2mm} \noindent {\bf Process Accept}. When a shard sequences an `accept($T$, commit)' decision, it sets all objects that are inputs to the transaction $T$ as being inactive (\Cref{fig:state}). It also creates any output objects from $T$ via BFT consensus that are to be managed by the shard. If the output objects are not managed by the shard, the shard sends requests to the concerned shards to create the objects. 
    On the other hand if the shard decision is `accept($T$, abort)',  all nodes release locks held on inputs or references of transaction $T$. Thus those objects remain active and may be used by other transactions.
    %\item  Upon reaching consensus for either `accept($T$, commit)' or `accept($T$, abort)` for the first time, all nodes in a shard broadcast the decision `create(T)' to all concerned peers in shards managing output objects of the transaction $T$. They attach all the signed evidence leading to this decision (all signed `prepared(commit) messages)'. Furthermore, in case of a `commit' used objects are considered inactive for subsequent actions, and the decision to accept T is broadcast to all shards managing output objects of $T$. In case of an `abort' all locks for objects used by $T$ are released.
    %\item Any subsequent `accept($T$, *)' action for transaction $T$ is ignored by peers. Thus on abort, due to locks, a user has to re-package a new transaction and retry.
    % \item When an honest peer in a shard receives a `create(T)' message, it schedules it within the BFT protocol. Once sequenced all peers check the validity of the evidence associated with the message, and create the objects output by $T$ managed by the shard.

As previously mentioned, some of the messages in \sbac are handled by a designated node in each shard called the \bftinit. Specifically, the \bftinit drives the composed \sbac protocol by sending `prepare(T)' and then `accept($T$, *)' messages to reach BFT consensus within and across shards. It is also responsible for broadcasting consensus decisions to relevant parties. The protocol supports a two-phase process to recover from a malicious \bftinit that suppresses transactions. As nodes in a shard hear all messages, they wait for the \bftinit to act on it until they time out. They first send a reminder to the \bftinit along with the original message to account for network losses. Next they proceed to wait; if they time out again, other nodes perform the action of \bftinit which is idempotent.        

\subsection{Concurrency \& Scalability}

% Transactions are packaged by \sysname clients and submitted to one or more nodes in the system. Those transactions include all the input and referenced objects concerning the transaction, any local state contributing to the transaction, and the resulting new objects. % \sysname peers execute the the transaction checker for the smart contract linked to the transaction, to ensure its validity, and then ensure (1) that all the objects input or referenced are valid and (2) that there is consensus that the transaction is the only to be executed on those valid objects, before accepting it. If those checks succeed the transaction is accepted and its input objects turned inactive.

Each transaction $\transaction$ involves a fixed number of \emph{concerned nodes} $\Phi(\transaction)$ within \sysname, corresponding to the shards managing its inputs and references. If two transactions $\transaction_0$ and $\transaction_1$ have disjoint sets of concerned nodes ($\Phi(T_0) \cap \Phi(T_1) = \emptyset$) they cannot conflict, and are executed in parallel or in any arbitrary order. % Even if this set is not empty, but instead the input objects of those two transactions are disjoint, this property holds.
If however, two transactions have common input objects, only one of them is accepted by all nodes. This is achieved through the \sbac protocol. It is local, in that it concerns only nodes managing the conflicting transactions, and does not require a global consensus. % Furthermore, the nodes managing the conflicting input are guaranteed to overlap between the two conflicting transactions.  and ensuring a quorum of peers managing each object agrees to accept a new transaction, would guarantee that no conflicting transactions will be accepted.

From the point of view of scalability, \sysname capacity grows linearly as more shards are added, subject to transactions having on average a constant, or sub-linear, number of inputs and references (see \Cref{fig:implementation1}). Furthermore, those inputs must be managed by different nodes within the system to ensure that load of accepting transactions is distributed across them. 

%How such distribution is achieved depends on the threat model of the application. Some may opt for a totally peer to peer model, where an ad-hoc random quorum of peers manages each object. Other application may opt for a small set of well-known authorities managing peers, with each object being managed by at least one representative peer from each authority.

\section{Security and Correctness}
\label{theorems}

\subsection{Security \& Correctness of \sbac}

The \sbac protocol guarantees a number of key properties, on which rest the security of \sysname, namely \emph{liveness} \emph{consistency}, and \emph{validity}. Before proceeding with stating those properties in details, and proving them we note three key invariants, that nodes may decide:
\begin{itemize}
\item \LP(commit / abort, T): A node considers that \LP(commit / abort, T) for a shard holds, if it receives at least $\faulty + 1$ distinct signed messages from nodes in the shard, stating `prepared(commit, T)' or `prepared(abort, T)' respectively. As a special case a node automatically concludes \LP(commit / abort, T) for a shard it is a member of, if all the preconditions necessary to provide that answer are present when an `prepare(T)' is sequenced.
\item \APC(commit, T): A node considers that `\APC(commit, T)' holds if it believes that `\LP(commit, T)' holds for all shards with concerned nodes for $T$. Note this may only be decided after reaching a conclusion (e.g.\ through receiving signed messages) about all shards.
\item \SPA(abort, T): A node considers that `\SPA(abort, T)' holds if it believes that `\LP(abort, T)' holds for at least one shard with concerned nodes for $T$. This may be concluded after only reaching a conclusion for a single shard, including the shard the node may be part of.
\end{itemize}

Liveness ensures that transactions make progress once proposed by a user, and no locks are held indefinitely on objects, preventing other transactions from making progress. 
\begin{sbacThm}
Liveness: Under the `honest shards' threat model, a transaction $\transaction$ that is proposed to at least one honest concerned node, will eventually result in either being committed or aborted, namely all parties deciding accept(commit, T) or accept(abort, T). 
\end{sbacThm}
\begin{proof}
\small
We rely on the liveness properties of the byzantine agreement (shards with only $\faulty$ nodes will reach a consensus on a sequence), and the broadcast from nodes of shards to all other nodes of shards, including the shards that manage transaction outputs. Assuming prepare(T) has been given to an honest node, it will be sequenced withing an honest shard BFT sequence, and thus a prepared(commit, T) or prepared(abort, T) will be sent from the $2\faulty+1$ honest nodes of this shard, to the $2\faulty+1$ nodes of the other concerned shards. Upon receiving these messages the honest nodes from other shards will schedule a prepare(T) message within their shards, and the BFT will eventually sequence it. Thus the user and all other honest concerned nodes will receive enough `prepared' messages to decide whether to proceed with `\APC(commit, T)' or `\SPA(abort, T)' and proceed with sequencing them through BFT. Eventually, each shard will sequence those, and decide on the appropriate `accept'.
\end{proof}

The second key property ensures that the execution of all transactions could be serialized, and thus is correct.

\begin{sbacThm}
Consistency: Under the `honest shards' threat model, no two conflicting transactions, namely transactions sharing the same input will be committed. Furthermore, a sequential executions for all transactions exists. 
\end{sbacThm}
\begin{proof}
\small
A transaction is committed only if some nodes conclude that `\APC(commit, T)', which presupposes all shards have provided enough evidence to conclude `\LP(commit, T)' for each of them. Two conflicting transaction, sharing an input or reference, must share a shard of at least $3\faulty+1$ concerned nodes for the common object---with at most $\faulty$ of them being malicious. Without loss of generality upon receiving the prepare(T) message for the first transaction, this shard will sequence it, and the honest nodes will emit messages for all to conclude `\APC(commit, T)'---and will lock this object until the two phase protocol concludes. Any subsequent attempt to prepare(T') for a conflicting T' will result in a \LP(abort, T') and cannot yield a commit, if all other shards are honest majority too. After completion of the first `accept(commit, T)' the shard removes the object from the active set, and thus subsequent T' would also lead to \SPA(abort, T'). Thus there is no path in the chain of possible interleavings of the executions of two conflicting transactions that leads to them both being committed.
\end{proof}

\begin{sbacThm}
Validity: Under the `honest shards' threat model, a transaction may only be committed if it is valid according to the smart contract checkers matching the traces of the procedures it executes. 
\end{sbacThm}
\begin{proof}
\small
A transaction is committed only if some nodes conclude that `\APC(commit, T)', which presupposes all shards have provided enough evidence to conclude `\LP(commit, T)' for each of them. The concerned nodes include at least one shard per input or reference object for the transaction; for any contract $\contract$ represented in the transaction, at least one of those shards will be managing object from that contract. Each shard checks the validity rules for the objects they manage (ensuring they are active, and not locked) and the contracts those objects are part of (ensuring the calls to $\contract$ pass its checker) in order to \LP(accept, T). Thus if all shards say \LP(accept, T) to conclude that `\APC(commit, T)', all object have been checked as active, and all the contract calls within the transaction have been checked by at least one shard---whose decision is honest due to at most $\faulty$ faulty nodes. If even a single object is inactive or locked, or a single trace for a contract fails to check, then the honest nodes in the shard will emit `prepared(abort, T)' upon sequencing `prepare(T)', and the final decision will be `\SPA(abort, T)'.
\end{proof}

\subsection{Auditability}

In the previous sections we show that if each shard contains at most $\faulty$ faulty nodes (honest shard model), the \sbac protocol guarantees consistency and validity. In this section we argue that if this assumption is violated, i.e.\ one or more shards contain more than $\faulty$\ byzantine nodes each, then honest shards can detect faulty shards. Namely, enough auditing information is maintained by honest nodes in \sysname to detect inconsistencies and attribute them to specific shards (or nodes within them).

The rules for transaction validity are summarized in \Cref{fig:rules}. Those rules are checked in a distributed manner: each shard keeps and checks the active or inactive state of objects assigned to it; and also only the contract checkers corresponding to the type of those objects. An honest shard emits a proposed(T, commit) for a transaction T only if those checks pass, and proposed(T, abort) otherwise or if there is a lock on a relevant object. A dishonest shard may emit proposed(T, *) messages arbitrarily without checking the validity rules. By definition, an invalid transaction is one that does not pass one or more of the checks defined in \Cref{fig:rules} at a shared, for which the shard has erroneously emitted a proposed(T, commit) message. 

\begin{SecThm}
Auditability: A malicious shard (with more than \faulty\ faulty nodes) that attempts to introduce an invalid transaction or object into the state of one or more honest shards, can be detected by an auditor performing a full audit of the \sysname system.
\end{SecThm}
\begin{proof}
\small
We consider two hash-chains from two distinct shards. We define the pair of them as being valid if (1) they are each valid under full audit, meaning that a re-execution of all their transactions under the messages received yields the same decisions to commit or abort all transactions; and (2) if all prepared(T,*) messages in one chain are compatible with all messages seen in the other chain. In this context `compatible' means that all prepared(T,*) statements received in one shard from the other represent the `correct' decision to commit or abort the transaction T in the other shard. An example of incompatible message would result in observing a proposed(T, commit) message being emitted from the first shard to the second, when in fact the first shard should have aborted the transaction, due to the checker showing it is invalid or an input being inactive.

Due to the property of digital signatures (unforgeability and non-repudiation), if two hash-chains are found to be `incompatible', one belonging to an honest shard and one belonging to a dishonest shard, it is possible for everyone to determine which shard is the dishonest one. To do so it suffices to isolate all statements that are signed by each shard (or a peer in the shard)---all of which should be self-consistent. It is then possible to show that within those statements there is an inconsistency---unambiguously implicating one of the two shards in the cheating. Thus, given two hash-chains it is possible to either establish their consistency, under a full audit, or determine which belongs to a malicious shard. 
\end{proof}

Note that the mechanism underlying tracing dishonest shards is an instance of the age-old double-entry book keeping\footnote{The first reported use is 1340AD~\cite{lauwers1994five}.}: shards keep records of their operations as a non-repudiable signed hash-chain of checkpoints---with a view to prove the correctness of their operations. They also provide non-repudiable statements about their decisions in the form of signed proposed(T,*) statements to other shards. The two forms of evidence must be both correct and consistent---otherwise their misbehaviour is detected.

\section{System and Applications Smart Contracts}
\label{applications}
\subsection{System Contracts} \label{System Contracts}

The operation of a \sysname distributed ledger itself requires the maintenance of a number of high-integrity high-availability data structures. Instead of employing an ad-hoc mechanism, \sysname employs a number of \emph{system smart contracts} to implement those.  Effectively, instantiation of \sysname is the combination of nodes running the basic \sbac protocol, as well as a set of system smart contracts providing flexible policies about managing shards, smart contract creation, auditing and accounting. This section provides an overview of system smart contracts.

\vspace{2mm} \noindent {\bf Shard management.} The discussion of \sysname so far, has assumed a function $\shard(\csobject)$ mapping an object $\csobject$ to nodes forming a shard. However, how those shards are constituted has been abstracted. A smart contract \textsf{ManageShards} is responsible for mapping nodes to shards. \textsf{ManageShards} initializes a singleton object of type \textsf{MS.Token} and provides three procedures: \textsf{MS.create} takes as input a singleton object, and a list of node descriptors (names, network addresses and public verification keys), and creates a new singleton object and a \textsf{MS.Shard} object representing a new shard; \textsf{MS.update} takes an existing shard object, a new list of nodes, and $2\faulty+1$ signatures from nodes in the shard, and creates a new shard object representing the updated shard. Finally, the \textsf{MS.object} procedure takes a shard object, and a non-repudiable record of malpractice from one of the nodes in the shard, and creates a new shard object omitting the malicious shard node---after validating the misbehaviour. Note that \sysname is `open' in the sense that any nodes may form a shard; and anyone may object to a malicious node and exclude it from a shard.

\vspace{2mm} \noindent {\bf Smart-contract management.} \sysname is also `open' in the sense that anyone may create a new smart contract, and this process is implemented using the \textsf{ManageContracts} smart contract. \textsf{ManageContracts} implements three types: \textsf{MC.Token}, \textsf{MC.Mapping} and \textsf{MC.Contract}. It also implements at least one procedure, \textsf{MC.create} that takes a binary representing a checker for the contract, an initialization procedure name that creates initial objects for the contract, and the singleton token object. It then creates a number of outputs: one object of type \textsf{MC.Token} for use to create further contracts; an object of type \textsf{MC.Contract} representing the contract, and containing the checker code, and a mapping object \textsf{MC.mapping} encoding the mapping between objects of the contract and shards within the system. Furthermore, the procedure \textsf{MC.create} calls the initialization function of the contract, with the contract itself as reference, and the singleton token, and creates the initial objects for the contract.

Note that this simple implementation for \textsf{ManageContracts} does not allow for updating contracts. The semantics of such an update are delicate, particularly in relation to governance and backwards compatibility with existing objects. We leave the definitions of more complex, but correct, contracts for managing contracts as future work. In our first implementation we have hardcoded \textsf{ManageShards} and \textsf{ManageContracts}.

\vspace{2mm} \noindent {\bf Payments for processing transactions.} \sysname is an open system, and requires protection againt abuse resulting from overuse. To achieve this we implement a method for tracking value through a contract called \textsf{CSCoin}.

The \textsf{CSCoin} contract creates a fixed initial supply of coins---a set of objects of type The \textsf{CSCoin.Account} that may only be accessed by a user producing a signature verified by a public key denoted in the object. A \textsf{CSCoin.transfer} procedure allows a user to input a number of accounts, and transfer value between them, by producing the appropriate signature from incoming accounts. It produces a new version of each account object with updated balances. This contract has been implemented in Python with approximately 200 lines of code.

The \textsf{CSCoin} contract is designed to be composed with other procedures, to enable payments for processing transactions. The transfer procedure outputs a number of local returns with information about the value flows, that may be used in calling contracts to perform actions conditionally on those flows. Shards may advertise that they will only consider actions valid if some value of \textsf{CSCoin} is transferred to their constituent nodes. This may apply to system contracts and application contracts.

\subsection{ Application level smart contracts}
This section describes some examples of privacy-friendly smart contracts and showcases how smart contract creators may use \sysname to implement advanced privacy mechanisms.

\iffalse
% Sensor
\vspace{2mm} \noindent {\bf Sensor---`Hello World' contract}
To illustrate \sysname's applications, we implement a simple 150 lines contract aggregating data from different sensors, called \textsf{Sensor}. This contract defines the types \textsf{Sensor.Token} and \textsf{Sensor.Data}, and two procedures  \textsf{Sensor.createSensor} and \textsf{Sensor.addData}. The procedure \textsf{Sensor.createSensor} takes as input the singletone token (that is created upon contract creation), and outputs a fresh \textsf{Sensor.Data} object with initially no date. The \textsf{Sensor.addData} procedure is applied on a \textsf{Sensor.Data} object with some new sensor's data as parameter, and a new object \textsf{Sensor.Data} appending the list of new data to the previous ones is created.
\fi

% Smart metering
\vspace{2mm} \noindent {\bf Smart-Meter Private Billing.} 
%A body of work~\cite{DBLP:conf/pet/JawurekJK11,DBLP:conf/isse/RialD12} examines how to achieve privacy-friendly time of use billing for smart meter deployments---a use-case requiring both high-integrity, and also privacy. Thus it showcases how smart contract creators may use \sysname to implement advanced privacy mechanisms.
%%Achieving privacy-friendly smart meter deployments requires both high-integrity, and also privacy. Thus it showcases how smart contract creators may use \sysname to implement advanced privacy mechanisms.

We implement a basic private smart-meter billing mechanism~\cite{DBLP:conf/pet/JawurekJK11,DBLP:conf/isse/RialD12} using the contract \textsf{SMet}: it implements three types \textsf{SMet.Token}, \textsf{SMet.Meter} and \textsf{SMet.Bill}; and three procedures, \textsf{SMet.createMeter}, \textsf{SMet.AddReading}, and \textsf{SMet.computeBill}. The procedure \textsf{SMet.createMeter} takes as input the singletone token and a public key and signature as local parameters, and it outputs a \textsf{SMet.Meter} object tied to this meter public key if the signature matches. \textsf{SMet.Meter} objects represent a collection of readings and some meta-data about the meter. Subsequently, the meter may invoke \textsf{SMet.addReading} on a \textsf{SMet.Meter} with a set of cryptographic commitments readings and a period identifier as local parameters, and a valid signature on them. A signature is also included and checked to ensure authenticity from the meter. A new object \textsf{SMet.Meter} is output appending the list of new readings to the previous ones. Finally, a procedure \textsf{SMet.computeBill} is invoked with a \textsf{SMet.Meter} and local parameters a period identifier, a set of tariffs for each reading in the period, and a zero-knowledge proof of correctness of the bill computation. The procedure outputs a \textsf{SMet.Bill} object, representing the final bill in plain text and the meter and period information.

%The \textsf{SMet} contract illustrates the need to separate procedures from checkers for contracts: the procedures behind \textsf{SMet.addReadings} for example takes secret readings as parameters, as well as commitments to them. %that leak no information that are included in the \textsf{SMet.Meter} object. 
%%But the checker, does not have access to the secret reading, and neither anyone observing the \textsf{SMet.Meter} object contents. Similarly, the \textsf{SMet.computeBill} procedure takes again as secret parameters the plaintext readings, and computes potentially the zero-knowledge proof of correctness of the \textsf{SMet.Bill} object. 
This proof of correctness is provided to the checker---rather than the secret readings---which proves that the readings matching the available commitments and the tariffs provided yield the bill object. The role of the checker, which checks public data, in both those cases is very different from the role of the procedure that is passed secrets not available to the checkers to protect privacy. 
%%The \textsf{SMet} contract also illustrates the advantages of using an open distributed ledger to store commitments to readings: they are available for the household owner to view, and the \textsf{SMet.computeBill} procedure may be used to compute any weighted sums of readings---which may be used to bill for consumption%, receive payments for energy produced, produce statistical aggregates of household consumption to receive green-energy related rebates---and other contracts may use \textsf{SMet} as a sub-system to achieve those. 
%%In particular a contract can combine calls to \textsf{SMet} and \textsf{CSCoin} to privately compute bills and perform payments at the same time. 
This contracts has been implemented in about 200 lines of Python code and is evaluated in section \Cref{evaluation}.

% Vote
\vspace{2mm} \noindent {\bf A Platform for Decision Making.}
An additional example of \sysname's privacy-friendly application is a smart voting system. We implement the contract \textsf{SVote} with three types, \textsf{SVote.Token}, \textsf{SVote.Vote} and \textsf{SVote.Tally}; and three procedures. 

\textsf{SVote.createElection}, consumes a singleton token and takes as local parameters the options, a list of all voter's public key, the tally's public key, and a signature on them from the tally. It outputs a fresh \textsf{SVote.Vote} object, representing the initial stage of the election (all candidates having a score of zero) along with a zero-knowledge proof asserting the correctness of the initial stage.

\textsf{SVote.addVote}, is called on a \textsf{SVote.Vote} object and takes as local parameters a new vote to add, homomorphically encrypted and signed by the voter. In addition, the voter provides a zero-knowledge proof certifying that her vote is a binary value and that she voted for exactly one option. The voter's public key is then removed from the list of participants to ensure that she cannot vote more than once. If all proofs are verified by the checker and the voter's public key appears in the list, a new \textsf{SVote.Vote} object is created as the homomorphic addition of the previous votes with the new one. Note that the checker does not need to know the clear value of the votes to assert their correctness since it only has to verify the associated signatures and zero-knowledge proofs.

Finally, the procedure \textsf{SVote.tally} is called to threshold decrypt the aggregated votes and provide a \textsf{SVote.Tally} object representing the final election's result in plain text, along with a proof of correct decryption from the tally. The \textsf{SVote} contract's size is approximately 400 lines.

%\george{TODO: benchmarks \& composition.}

%\george{What we could have:}

%\vspace{2mm} \noindent {\bf Certificate Transparency}

% \subsection{Music Rights Tracking}

%\vspace{2mm} \noindent {\bf A Distributed Autonomous Organization}

%Instantiating DECODE Applications

%\vspace{2mm} \noindent {\bf A Platform for Civic Deliberation \& Decision Making}

%\vspace{2mm} \noindent {\bf A collaborative Spare Room sharing Platform}

%\vspace{2mm} \noindent {\bf Private banking}

%\vspace{2mm} \noindent {\bf A Distributed Local Exchange Market}

%The likelihood is that most of the above cases have two variants, that we discuss in more details when describing how to integrate privacy into the \sysname system: an \emph{access control} variant where the data in the chain is not publicly available or encrypted, and is verifiable for an auditor looking at the chain with data in the clear; a \emph{zero-knowledge} variant where the data in the chain is encrypted and publicly available, and is verifiable by anyone. This comes at the cost of additional complexity for incorporating zero-knowledge proofs into the checkers.

\section{Implementation \& Evaluation}
\label{evaluation}

We implemented a prototype of \sysname in $\sim$10K lines of Python and Java code. The implementation consists of two components: a Python contracts environment and a Java node. We have released the code as an open-source project on GitHub.\footnote{URL omitted for double-blind review.}

\vspace{2mm} \noindent {\bf Python Contract Environment.}
The Python contracts environment allows developers to write, deploy and test smart contracts. These are deployed on each node by running the Python script for the contract, which starts a local web service for the contract's checker. The contract's checker is then called though the web service. The environment provides a framework to allow developers to write smart contracts with little worry about the underlying implementation, and provides an auto-generated checker for simple contracts.

\vspace{2mm} \noindent {\bf Java Node Implementation.}
The Java node implements a shard replica that accepts incoming transactions from clients and initiates, and executes, the \sbac protocol. For BFT consensus, we use the \bftsmart~\cite{bftsmart} Java library---one of the very few maintained open source libraries implementing byzantine consensus.

\begin{figure}[!t]
    \centering
    \includegraphics[width=.48\textwidth]{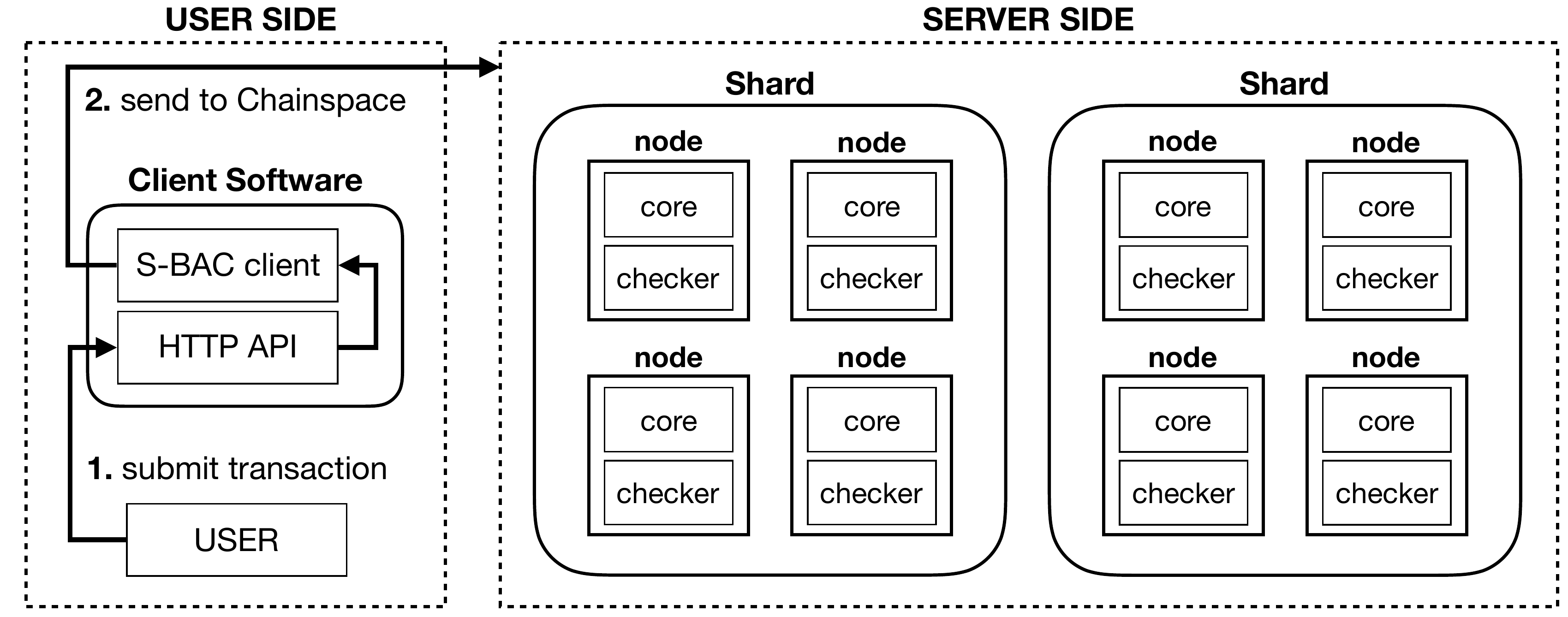}
    \caption{Diagram illustrating the implementation of a \sysname system with two shards managing four nodes each. The user submits the transaction to its local \sbac client through a built-in HTTP API (arrow 1). Then, this \sbac client sends the transaction to \sysname (arrow 2).}
    \label{fig:implementation}
\end{figure}

To communicate with \sysname, end users also run an \sbac--enabled client. First, she creates a transaction through the Python environments using one or many existing smart contracts.
%\begin{wrapfigure}{r}{0.2\textwidth}
%  \begin{center}
%    \includegraphics[width=.15\textwidth]{impl2}
%  \end{center}
%  \caption{Each node is divided into two part: a client and a server. The server's core is the transaction processing part receiving the transaction (arrow 1), that calls the checker to verify the transaction's validity (arrow 2).}
%  \label{fig:server}
%\end{wrapfigure}
She then submits the transaction to its \sbac client through the HTTP API as indicated in \Cref{fig:implementation}, that sends the transaction to \sysname according to the \bftsmart   protocol.
 
A node is composed of a server divided in two parts: the core and the checker. To communicate with other nodes, each node also contains an \sbac client. When a transaction is received, the core is in charge of verifying that the input objects and references are active (neither locked nor inactive).  Then, the node runs the checker associated with the contract, in an isolated container. (The checker is provided by the contract's creator when the node starts up, and interfaces with the node through an HTTP API.) When the client submits a transaction with dependencies, the core recursively checks each dependent transaction first, and the top-level transaction at last (similar to depth-first search algorithm). 

\vspace{2mm} \noindent {\bf Performance Measurements.} We evaluated the performances and scalability of our implementation of \sysname, through deployments on Amazon EC2 containers. We launched up to 96 nodes on \emph{t2.medium} virtual machines, each containing 8 GB of RAM on 2 virtual CPUs and running GNU/Linux Debian 8.1. We sent transactions to the network from a \sysname client running on a t2.xlarge virtual machine, containing 16 GB of RAM and 4 virtual CPUs, also running GNU/Linux Debian 8.1. In our tests, we map objects to shards randomly using the mapping function $\phi(o) = id(o) \mod K$ where $K$ is a constant representing the number of shards and $id(o)$ is the SHA256 hash of the object.

We first measure the effect of the number of shards on transaction throughput (\Cref{fig:implementation1}). The transaction throughput of \sysname scales linearly with the number of shards: with 4 nodes per shard, the number of transactions per second (t/s) increases on average by 22 for 1-input transactions for each shard added. This is because as inputs are randomly assigned to shards based on their hashes, the transaction processing load is spread out over a larger number of shards.

Next we investigate the effect of shard size (the number of nodes per shard) on transaction throughput (\Cref{fig:implementation4}). We fix the number of shards to 2, and increase the number of nodes per shard from 2 to 48. With \bftsmart  configured for $3f+1$ fault tolerance, we observe an expected gracious decrease in transaction throughput: for each node added, the throughput reduces on average by 1.6 transactions per second. This is because in order for a \bftsmart  node to realise consensus for a message, it must receive a result from at least $f+1$ nodes. Thus, the bottleneck is the latency of the $f+1$th node with the highest response time. \mustafa{...}

Another factor that can potentially affect transaction throughput is the number of inputs per transaction: the more shards touched by the transaction inputs, the longer it will take to run \sbac among all the concerned shards. In \Cref{fig:implementation2}, we study how the number of inputs per transaction affects transaction throughput. We measure this for 5 shards, varying the number of inputs per transaction from 1 to 10, and the inputs are randomly mapped to shards as previously stated. The transaction throughput decreases asymptotically until it becomes stable at around 40 transactions per second. This is because \sbac's maximum time in processing transactions is capped at the time it takes to process transactions that touch all the 5 shards. Increasing the number of inputs does not further deteriorate the transaction throughput.

Finally, we measure the client-perceived latency---the time from when a client submits a transaction until it receives a decision about whether the transaction has been committed---under varying system loads expressed in terms of transactions received per second.
\Cref{fig:implementation3} shows the effect of transactions received by the system per second (all 1-input transactions) on client-perceived latency for 2 shards, each having 4 nodes. Recall from \Cref{fig:implementation1} that the average throughput for a \sysname system with similar configuration is 75 1-input transactions per second. Consequently, we observe in \Cref{fig:implementation4} that the increase in latency with varying system loads is smaller for 20~t/s--60~t/s (average 69~ms), but the values start to get bigger after 60~t/s (average 210~ms). This is when the system reaches its maximum transaction throughput, causing a backlog of transactions to be processed.

\begin{figure}[t]
    \centering
    \includegraphics[width=.48\textwidth]{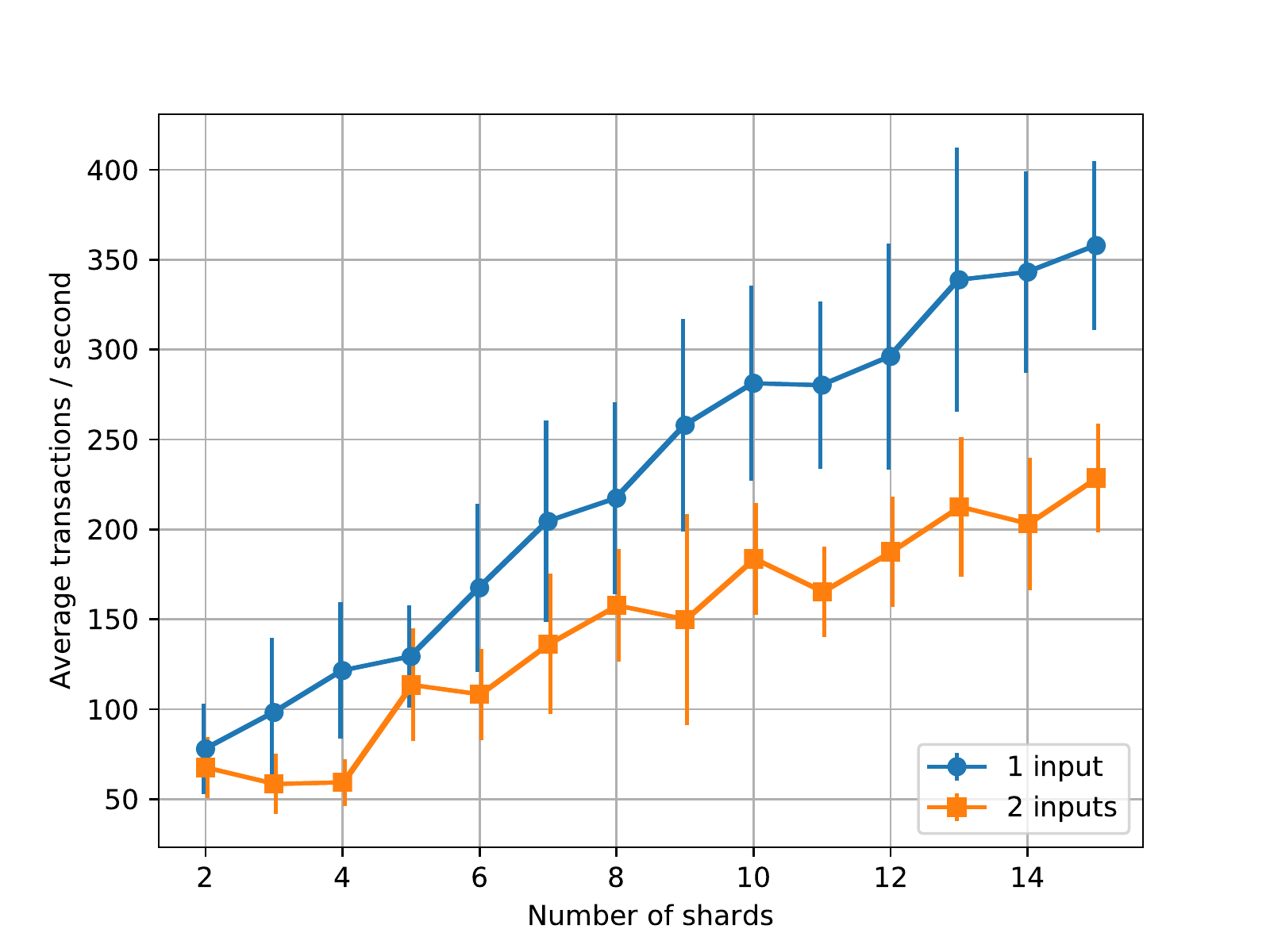}
    \caption{The effect of the number of shards on transaction throughput. (Shards: 2, nodes per shard: 4, input-to-shard mapping: random. Repeats: 20.)}
    \label{fig:implementation1}
\end{figure}

\begin{figure}[t]
    \centering
    \includegraphics[width=.48\textwidth]{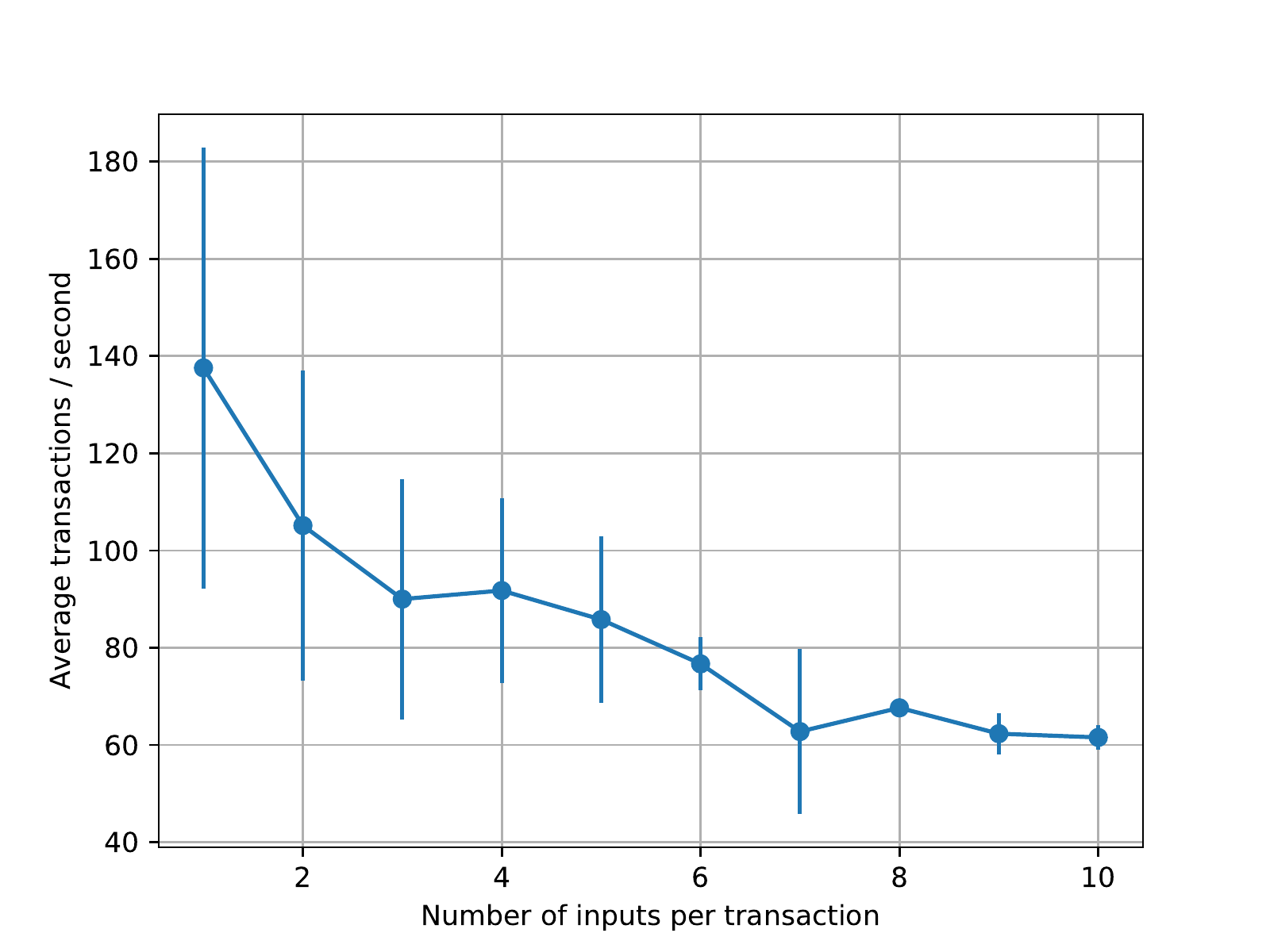}
    \caption{The effect of the number of inputs per transaction on transaction throughput. (Shards: 2, nodes per shard: 4, input-to-shard mapping: random. Repeats: 20.)}
    \label{fig:implementation2}
\end{figure}

\begin{figure}[t]
    \centering
    \includegraphics[width=.48\textwidth]{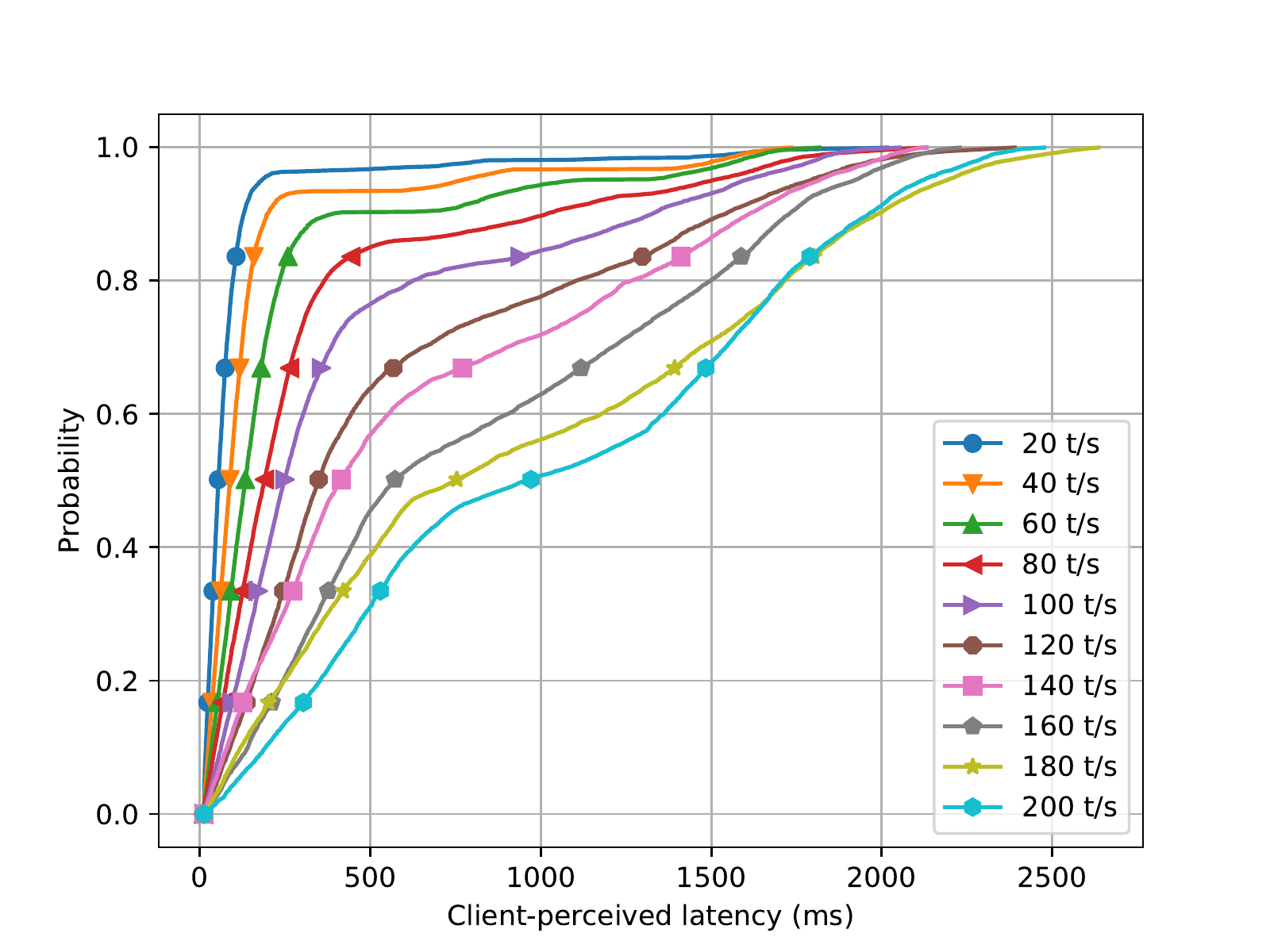}
    \caption{The cumulative distribution function of delay for the client to receive a final commit or abort response, for varying system load. (Shards: 2, nodes per shard: 4, inputs per transaction: 1, input-to-shard mapping: random. Repeats: 5.)}
    \label{fig:implementation3}
\end{figure}

\begin{figure}[t]
    \centering
    \includegraphics[width=.48\textwidth]{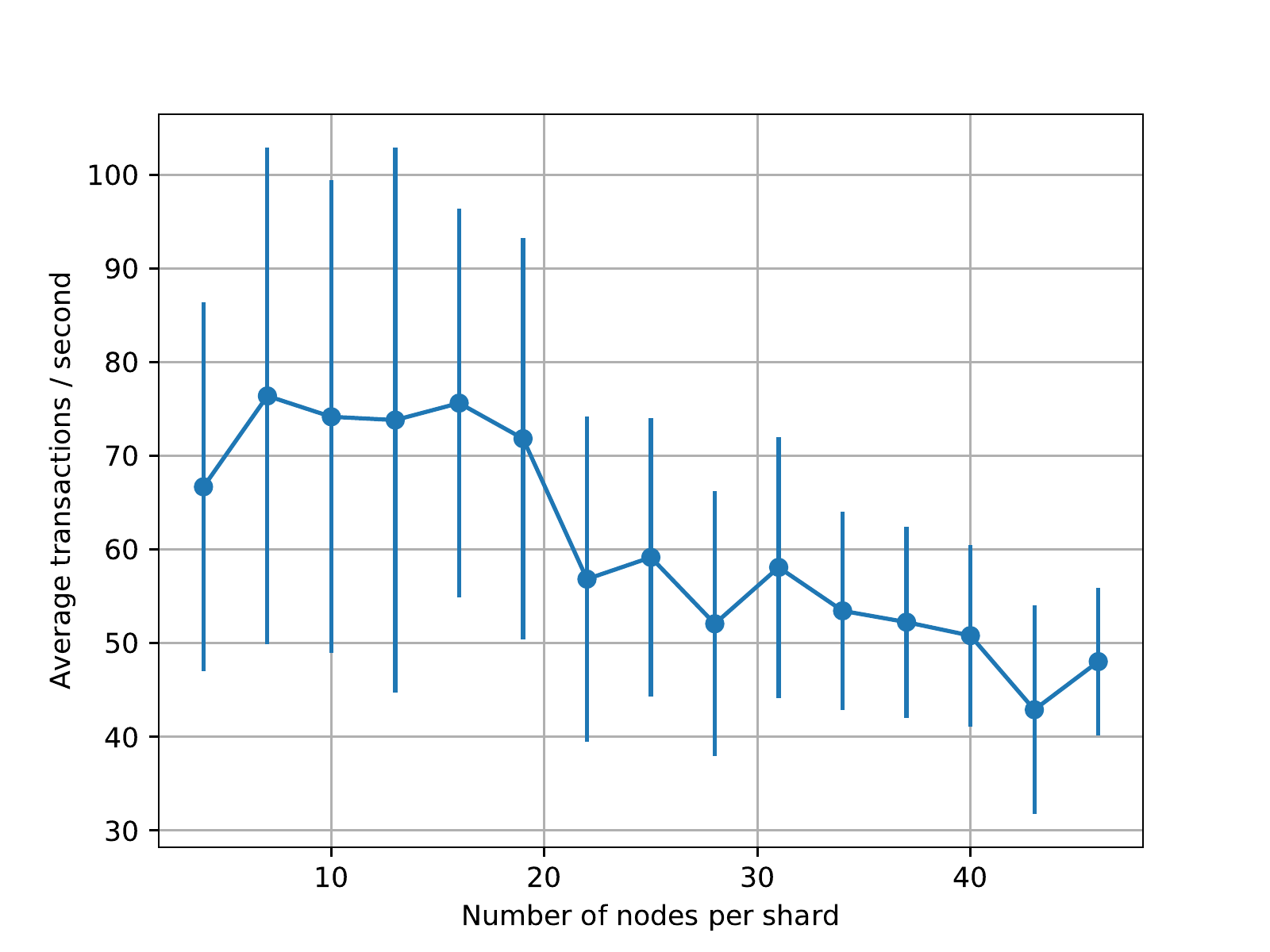}
    \caption{The effect of the number of nodes per shard on transaction throughput. (Shards: 2, inputs per transaction: 1, input-to-shard mapping: random. Repeats: 20.)}
    \label{fig:implementation4}
\end{figure}

\vspace{2mm} \noindent {\bf Smart Contract Benchmarks.}
We evaluate the cost and performance of some smart contracts described in \Cref{System Contracts}. We compute the mean ($\mu$) and standard deviation ($\sigma$) of the execution of each procedure (denoted as [g]) and checker (denoted as [c]) in the contracts. Each figure is the result of 10,000 measured on a dual-core Apple MacBook Pro 4.1, 2.7GHz Intel Core i7. The last column indicates the transaction's size resulting from executing the procedure. All cryptographic operations as digital signatures and zero-knowledge proofs have been implemented using the Python library petlib~\cite{petlib}, wrapping OpenSSL.

\iffalse
\begin{center}
    \footnotesize
    \tabulinesep=1pt
    \begin{tabu} to 0.47\textwidth {  X[1.5cm,r] X[c] X[c] X[c] X[r] }
        \multicolumn{3}{l}{ \textsf{Sensor}---Contract size: $\sim$150 lines}\\
        \multicolumn{1}{l}{Operation} & & \multicolumn{1}{c}{$\mu$ [ms]} & \multicolumn{1}{r}{$\sigma$ [ms]} & \multicolumn{1}{r}{size [B]}\\
        \hline
        \textsf{createSensor} & [g] & 0.139 & $\pm$ 0.039 & 416\\ 
                              & [c] & 0.021 & $\pm$ 0.008 & - \\ 
        \textsf{addData} & [g] & 0.105 & $\pm$ 0.085 & 417\\  
                         & [c] & 0.036 & $\pm$ 0.015 & - \\  
        \hline
    \end{tabu}
\end{center}
The \textsf{Sensor} contract is very cheap since it does not contain any cryptographic operation; the checker only has to verify the format of the transaction. The resulting micro benchmarks for \textsf{createSensor} is therefore a good indication of just the overhead of the \sysname system---to execute a procedure or a checker.
\fi

\vskip.3cm
\begin{center}
    \footnotesize
    \tabulinesep=1pt
    \begin{tabu} to 0.47\textwidth {  X[1.4cm,r] X[c] X[c] X[r] }
        \multicolumn{3}{l}{ \textsf{CSCoin}---Contract size: $\sim$200 lines}\\
        \multicolumn{1}{l}{Operation} & \multicolumn{1}{c}{$\mu$ [ms]} & \multicolumn{1}{r}{$\sigma$ [ms]} & \multicolumn{1}{r}{size [B]}\\
        \hline
        \textsf{createAccount} [g] & 4.845 & $\pm$ 0.683 & 512\\ 
        \textsf{ }[c] & 0.022 & $\pm$ 0.005 & -\\ 
        \textsf{authTransfer} [g] & 4.986 & $\pm$ 0.684 & 1114\\  
        \textsf{ }[c] & 5.750 & $\pm$ 0.474 & -\\  
        \hline
    \end{tabu}
\end{center}
The user needs to generate a signing key pair to create an account in the \textsf{CSCoin} contract, which takes about 5~ms. However, verifying the account creation only requires to check the transaction's format, and it is therefore very fast. Transferring money is a little more expensive due to the need to sign the amount transferred and the beneficiary, and verifying the signature in the checker.

\vskip.3cm
\begin{center}
    \footnotesize
    \tabulinesep=1pt
    \begin{tabu} to 0.47\textwidth {  X[1.2cm,r] X[c] X[c] X[r] }
        
        \multicolumn{3}{l}{ \textsf{SMet}---Contract size: $\sim$200 lines}\\
        \multicolumn{1}{l}{Operation} & \multicolumn{1}{c}{$\mu$ [ms]} & \multicolumn{1}{c}{$\sigma$ [ms]} & \multicolumn{1}{r}{size [B]}\\
        \hline
        \textsf{createMeter} [g] & 4.786 & $\pm$ 0.480 & $\sim$600\\ 
        \textsf{} [c] & 0.060 & $\pm$ 0.003 & -\\ 
        \textsf{addReading} [g] & 5.286 & $\pm$ 0.506 & $\sim$1100\\  
        \textsf{} [c] & 5.965 & $\pm$ 0.697 & -\\  
        \textsf{computeBill} [g] & 5.043 & $\pm$ 0.513 & $\sim$1100\\  
        \textsf{} [c] & 5.870 & $\pm$ 0.603 & -\\  
        \hline
    \end{tabu}
\end{center}
Similarly to \textsf{CSCoin}, creating a meter requires generating a cryptographic key pair which takes about 5~ms, while verifying the meter's creation is faster and only requires checking the transaction's format. Adding new readings takes about 5~ms, as the user needs to create a signed commitment of the readings which requires elliptic curve operations and an ECDSA signature. Computing the bill takes slightly longer (5.8~ms), and involves homomorphic additions, and verifying the bill involves checking a zero-knowledge proof of the billing calculation.

\vskip.3cm
\begin{center}
    \footnotesize
    \tabulinesep=1pt
    \begin{tabu} to 0.47\textwidth {  X[1.4cm,r] X[c] X[c] X[r] }
        \multicolumn{3}{l}{ \textsf{SVote}---Contract size: $\sim$400 lines}\\
        \multicolumn{1}{l}{Operation} & \multicolumn{1}{c}{$\mu$ [ms]} & \multicolumn{1}{r}{$\sigma$ [ms]} & \multicolumn{1}{r}{size [B]}\\
        \hline 
        \textsf{createElection} [g] & 11.733 & $\pm$ 1.028 & $\sim$1227\\ 
        \textsf{} [c] & 11.327 & $\pm$ 0.782 & -\\ 
        \textsf{addVote} [g] & 14.086 & $\pm$ 1.043 & $\sim$2758\\  
        \textsf{} [c] & 28.178 & $\pm$ 1.433 & -\\ 
        \textsf{tally} [g] & 253.286 & $\pm$ 7.793 & $\sim$1264\\  
        \textsf{} [c] & 11.589 & $\pm$ 0.937 & -\\ 
        \hline
    \end{tabu}
\end{center}
The \textsf{SVote} contract is more expensive than the others since it extensively uses zero-knowledge proofs and more advanced cryptography. For simplicity, this smart contract has been tested with three voters and two options. First of all, creating a new election event requires building a signed homomorphic encryption of the initial value for each option, and a zero-knowledge proof asserting that the encrypted value is zero; this takes roughly 11~ms to generate the transaction and to run the checker. Next, each time a vote is added, the user proves two zero-knowledge statements---one asserting that she votes for exactly one option and one proving that her vote is a binary value---and computes an ECDSA signature on her vote, which takes about 11~ms and generates a transaction of about 2.7~kB. Verifying the signature and the two zero-knowledge proofs are slower and takes about 30~ms. Finally, tallying is the slowest operation since it requires to decrypt the homomorphic encryption of the votes' sum.

\section{Limitations}
\label{limits}
\sysname has a number of limitations, that are beyond the scope of this work to 
tackle, and deferred to future work.

The integrity properties of \sysname rely on all shards managing objects being honest, namely containing at most $\faulty$  fault nodes each. We have chosen to let any set of nodes can create a shard. However, this means that the function $\phi(o)$ mapping objects to shards must avoid dishonest shards. Our isolation properties ensure that a dishonest shard can at worse affect state from contracts that have objects mapped to it. Thus, in \sysname, we opt to allow the contract creator to designate which shards manage objects from their contract. This embodies specific trust assumptions where users have to trust the contract creator both for the code (which is auditable) and also for the choice of shards to involve in transactions---which is also public.
    
In case one or more shards are malicious, we provide an auditing mechanism for honest nodes in honest shards to detect 
the inconsistency and to trace the malicious shard. Through the Hash-DAG structure it is also possible to fully audit the histories of two objects, and to ensure that the validity rules hold jointly---in particular the double-use rules. However, it is not clear how to automatically recover from detecting such an inconsistency. Options include: forcing a fork into one or many consistent worlds; applying a rule to collectively agree the canonical version; patching past transactions to recover consistency; or agree on a minimal common consistent state. Which of those options is viable or best is left as future work.

Checkers involved in validating transactions can be costly. For this reason we allow peers in a shard to accept transactions subject to a \textsf{SCCoin} payment to the peers. However, this `flat' fee is not dependent on the cost or complexity of running the checker which might be more or less expensive. Etherium~\cite{wood2014ethereum} instead charges `gas' according to the cost of executing the contract procedure---at the cost of implementing their own virtual machine and language.
    
Finally, the \sbac protocol ensures correctness in all cases. However, under high contention for the same object the rate of aborted transactions rises. This is expected, since the \sbac protocol in effect implements a variant of optimistic concurrency control, that is known to result in aborts under high contention. There are strategies for dealing with this in the distributed systems literature, such as locking objects in some conventional order---however none is immediately applicable to the byzantine setting.

\section{Comparisons with Related Work}
\label{related}

Bitcoin's underlying blockchain technology suffers from scalability issues: with a current block size of 1MB and 10 minute inter-block interval, throughput is capped at about 7 transactions per second, and a client that creates a transaction has to wait for about 10 minutes to confirm. In contrast, mainstream payment processing companies like Visa confirm transactions within a few seconds, and have high throughput of 2000 transactions per second on average, peaking up to 56,000 transactions per second~\cite{visa}. Reparametrization of \bitcoin---such as \bitcoinng---can improve this to a limited extent up to 27 transactions per second and 12 second latency, respectively~\cite{croman2016scaling}. More significant improvement requires a fundamental redesign of the blockchain paradigm. 

%There are two key parameters of blockchain scalability: transaction throughput (the maximum rate at which the blockchain can process transactions) and latency (time to confirm that a transaction is included in the blockchain). (While previous work has identified additional scalability metrics~\cite{croman2016scaling}, we believe the ones we focus on are bottleneck issues and more challenging from a research perspective.)  

The most comparable system to \sysname is \omniledger~\cite{omniledger}---that was developed concurrently---and provides a scalable distributed ledger for a cryptocurrency, and cannot support generic smart contracts. \omniledger assigns nodes (selected using a Sybil-attack resistant mechanism) into shards among which state, representing coins, is split. The node-to-shard assignment is done every epoch using a bias-resistant decentralized randomness protocol~\cite{randhound} to prevent an adversary from compromising individual shards. A block-DAG (Directed Acyclic Graph) structure is maintained in each shard rather than a single blockchain, effectively creating multiple blockchains in which consensus of transactions can take place in parallel. Nodes within shards reach consensus through the Practical Byzantine Fault Tolerant (PBFT) protocol~\cite{castro1999practical} with \byzcoin~\cite{byzcoin}'s modifications that enable $O(n)$ messaging complexity. In contrast, \sysname uses \bftsmart's PBFT implementation~\cite{modsmart} as a black box, and inherits its $O(n^2)$ messaging complexity---however, \bftsmart  can be replaced with any improved PBFT variant without breaking any security assumptions. 

Similar to \sysname, \omniledger uses an atomic commit protocol to process transactions across shards. However, it uses a different, client-driven approach to achieve it. To commit a transaction, the client first sends the transaction to the network. The leader of each shard that is responsible for the transaction inputs (input shard) validates the transaction and returns a proof-of-acceptance (or proof-of-rejection) to the client, and inputs are locked. To unlock those inputs, the client sends proof-of-accepts to the output shards, whose leaders add the transaction to the next block to be appended to the blockchain. In case the transaction fails the validation test, the client can send proof-of-rejection to the input shards to roll back the transaction and unlock the inputs. To avoid denial-of-service, the protocol assumes that clients are incentivized to proceed to the Unlock phase. Such incentives may exist in a cryptocurrency application, where coin owners only can spend them, but do not hold for a generalized platform like \sysname where objects may have shared ownership. Hence, \sysname's atomic commit protocol has the entire shard---rather than a single untrusted client---act as a coordinator. Other related works include improvements to Byzantine consensus for reduced latency and decentralization~\cite{tendermint,stellar,ripple}, but these do not support sharding.

% \george{I do not think we need this:} To mitigate the storage and pruning cost of maintaining the full ledger, \omniledger maintains a meta-ledger that has compact state-blocks that summarize shard states for every epoch. This enables a client to  verify proof-of-existence for a transaction (similar to our partial audit \bano{ref}), however a full audit is not possible and dishonest shards cannot be detected \bano{confirm}. 

\elastico~\cite{elastico} scales by partitioning nodes in the network into a hierarchy of committees, where each committee is responsible for managing a subset (shard) of transactions consistently through PBFT. A final committee collates sets of transactions received from committees into a final block and then broadcasts it. At the end of each epoch, nodes are reassigned to committees through proof-of-work. The block throughput scales up almost linear to the size of the network. However, \elastico cannot process multi-shard transactions. 

\rscoin~\cite{rscoin} is a permissioned blockchain. The central bank controls all monetary supply, while mintettes (nodes authorized by the bank) manage subsets of transactions and coins.
%such that a shard can potentially overlap across mintettes for security and reliability. 
%A mintette maintains information about outputs of the transactions it manages, whether these have been spent and if so in which transactions. 
Like \omniledger, communication between mintettes takes place indirectly, through the client---and also relies on the client to ensure completion of transactions. \rscoin has low communication overhead, and the transaction throughput scales linearly with the number of mintettes, but cannot support generic smart contracts.

Some systems improve transaction latency by replacing its probabilistic guarantees with strong consistency.
\byzcoin~\cite{byzcoin} extends \bitcoinng for high transaction throughput. 
%\bitcoinng has two kinds of blocks: keyblocks announce a new leader and includes proof-of-work, while microblocks are generated by a leader for duration of the epoch to be appended to the blockchain. \byzcoin modifies how keyblocks are generated; a consensus group, instead of a solo leader, generates a keyblock followed by microblocks. The consensus group is dynamically formed by a window of recent miners. Each miner has voting power proportional to the number of mining blocks it has in the current window, which is proportional to its hash power. When a miner finds solution to the puzzle, it becomes a member of the current consensus group and receives a share in the current window which moves one step forwards (ejecting the oldest miner). 
A consensus group is organized into a communication tree where the most recent miner (the leader) is at the root. The leader runs an $O(n)$ variant of PBFT (using CoSI) to get all members to agree on the next microblock. The outcome is a collective signature that proves that at least two-thirds of the consensus group members witnessed and attested the microblock. A node in the network can verify in $O(1)$ that a microblock has been validated by the consensus group. PeerConsensus~\cite{Decker:2016} achieves strong consistency by allowing previous miners to vote on blocks. A \textit{Chain Agreement} tracks the membership of identities in the system that can vote on new blocks. \algorand~\cite{algorand} replaces proof-of-work with strong consistency by proposing a faster \textit{graded} Byzantine fault tolerance protocol, that allows for a set of nodes to decide on the next block. A key aspect of Algorand is that these nodes are selected randomly using algorithimic randomness based on input from previously generated blocks. However, none of those systems are designed to support generic smart contracts.

Some recent systems provide a transparent platform based on blockchains for smart contracts.  Hyperledger Fabric~\cite{cachin2016architecture} is a permissioned blockchain to setup private infrastructures for smart contracts. It is designed around the idea of a `consortium' blockchain, where a specific set of nodes are designated to validate transactions, rather than random nodes in a decentralized network. Each smart contract (called \textit{chaincode}) has its own set of \textit{endorsers} that re-execute submitted transactions to validate them. A \textit{consensus service} then orders transactions and filters out those endorsed by too few. It uses \textit{modular consensus}, which is replaceable depending on the requirements (e.g., Apache Kafka or SBFT).

Ethereum~\cite{wood2014ethereum} provides a decentralized Turing-complete virtual machine, called EVM, able to execute smart-contracts. Its main scalability limitation results from every node having to process every transaction, as Bitcoin. On the other hand, \sysname's sharded architecture allows for a ledger linearly scalable since only the nodes concerned by the transaction---that is, managing the transaction's inputs or references---have to process it. Ethereum plans to improve scalability through sharding techniques~\cite{buterin2015notes}, but their work is still theoretical and does not provide any implementation or measurements. One major difference with \sysname is that Ethereum's smart contract are executed by the node, contrarily to the user providing the outputs of each transaction. \sysname also supports smart contracts written in any kind of language as long as checkers are pure functions, and there are no limitations for the code creating transactions. Some industrial systems~\cite{tezos, rootstock, corda} implement similar functionalities as \sysname, but without any empirical performance evaluation.

In terms of security policy, \sysname system implements a platform that enforces high-integrity by embodying a variant of the Clark-Wilson~\cite{clark1987comparison}, proposed before smart contracts were heard of. 

%Tezos is similar to \sysname in many aspects: it has a strong type checking system and implements its cryptocurrency as an "account" smart contract. However, Tezos's smart contracts are statefull and updateing a balance requires to rewrite the contract's storage space. This introduces many complications to prevent replay attacks and transaction's validity's check. \sysname avoid this situation by producing new version of account objects as specified in the Security Theorem 1. Moreover, Tezos implements an Ethereum-like gas system to pay nodes to execute the smart contract.

%Rootstock is a smart-contract platform that incorporate a Turing Complete virtual machine to Bitcoin. There are a number of Rootstock's smart-contracts illustrated in \cite{rootstock}, but without any scientific evidence of implementation. Rootstock is more scalable than Bitcoin since they use probabilistic verification and fraud proof as well as blockchain sharding techniques. \alberto{no proof of that, just a table with some numbers p17}

%Corda...

\section{Conclusions}
\label{conclusions}
We presented the design of \sysname---an open, distributed ledger platform for high-integrity and transparent processing of transactions. \sysname offers extensibility though privacy-friendly smart contracts. We presented an instantiation of \sysname by parameterizing it with a number of `system' and `application' contracts, along with their evaluation. However, unlike existing smart-contract based systems such as Ethereum~\cite{wood2014ethereum}, it offers high scalability through sharding across nodes using a novel distributed atomic commit protocol called \sbac, while offering high auditability. We presented implementation and evaluation of \sbac on a real cloud-based testbed under varying transaction loads and showed that \sysname's transaction throughput scales linearly with the number of shards by up to 22 transactions per second for each shard added, handling up to 350 transactions per second with 15 shards.  As such it offers a competitive alternative to both centralized and permissioned systems, as well as fully peer-to-peer, but unscalable systems like Ethereum.

% \george{Limit is 15 pages in TOTAL.}

% \balance

\ifproceedings{

\noindent {\bf Acknowledgements.} George Danezis, Shehar Bano and Alberto Sonnino are supported in part by EPSRC Grant EP/M013286/1 and the EU H2020 DECODE project under grant agreement number 732546. Mustafa Al-Bassam is supported by a scholarship from The Alan Turing Institute. Many thanks to Daren McGuinness and Ramsey Khoury for discussions about the \sysname design.
}\fi

\footnotesize

\bibliography{references}{}
\bibliographystyle{alpha}

% \appendix

%\section{Graveyard}
%\input{old}

\end{document}